\documentclass[journal,draftcls,onecolumn,12pt,twoside, narroweqnarray]{IEEEtranTCOM}

\ifCLASSINFOpdf
\else
\fi

\usepackage{algorithm}
\usepackage{algorithmic}

\makeatletter
\newcommand\fs@norules{\def\@fs@cfont{\bfseries}\let\@fs@capt\floatc@ruled
  \def\@fs@pre{}%
  \def\@fs@post{}%
  \def\@fs@mid{\kern3pt}%
  \let\@fs@iftopcapt\iftrue}
\makeatother
\floatstyle{norules}
\restylefloat{algorithm}

\usepackage{color}
\usepackage{algorithm}
\usepackage{algorithmic}
\usepackage{graphicx}
\usepackage{epstopdf}
\usepackage{mathrsfs}
\DeclareMathAlphabet{\mathpzc}{OT1}{pzc}{m}{it}

\newcounter{cc}
\usepackage{graphicx}
\usepackage{caption}
\usepackage{subcaption}

\usepackage{cite}
\usepackage{bbm}
\usepackage{float}

\usepackage{amsmath}
\usepackage{amsfonts}
\usepackage{amssymb}
\usepackage{amsthm}
\usepackage{slashbox}
\usepackage{bbding}
\usepackage{graphicx}
\usepackage{dsfont}
\usepackage{clrscode}
\usepackage[T1]{fontenc}
\usepackage{times}
\usepackage{mathptmx}
\usepackage[T1]{fontenc}

\usepackage{fancyhdr}
\definecolor{awesome}{rgb}{1.0, 0.13, 0.32}

\allowdisplaybreaks

\makeatletter
\newcommand{\rmnum}[1]{\romannumeral #1}
\newcommand{\Rmnum}[1]{\expandafter\@slowromancap\romannumeral #1@}
\makeatother

\usepackage{mathtools}
\DeclarePairedDelimiter\ceil{\lceil}{\rceil}
\DeclarePairedDelimiter\floor{\lfloor}{\rfloor}

\newtheorem{theorem}{Theorem}
\newtheorem{lemma}{Lemma}

\newtheorem{corollary}{Corollary}

\allowdisplaybreaks

\usepackage{setspace}

\begin{document}

\title{Downlink Non-Orthogonal Multiple Access with Limited Feedback}
\author{\IEEEauthorblockN{Xiaoyi (Leo) Liu, \textit{Student Member, IEEE}, Hamid Jafarkhani, \textit{Fellow, IEEE}} \\\
\thanks{X. Liu, H. Jafarkhani are with Center for Pervasive Communications \& Computing, University of California, Irvine, Irvine, CA 92697--2625 USA (Email:
\{xiaoyil3, hamidj\}@uci.edu).} }

\maketitle

\thispagestyle{plain}
\pagestyle{plain}

\begin{abstract}
In this paper, we analyze downlink non-orthogonal multiple access (NOMA) networks
with limited feedback. Our goal is to derive appropriate transmission rates for rate
adaptation and minimize outage probability of minimum rate for the constant-rate data
service, based on distributed channel feedback information from receivers. We propose an
efficient quantizer with variable-length encoding that approaches the best performance of
the case where perfect channel state information is available everywhere. We prove that in
the typical application with two receivers, the losses in the minimum rate and outage
probability decay at least exponentially with the minimum feedback rate. We analyze the
diversity gain and provide a sufficient condition for the quantizer to achieve the maximum
diversity order. For NOMA with $K$ receivers where $K > 2$, we solve the minimum rate maximization problem within an accuracy of $\epsilon$ in time complexity
of $O\left(K\log\frac{1}{\epsilon}\right)$, then, we apply the previously proposed
quantizers for $K = 2$ to the case of $K > 2$. Numerical simulations are presented to demonstrate
the efficiency of our proposed quantizers and the accuracy of the analytical results.

\end{abstract}


\IEEEpeerreviewmaketitle

\section{Introduction}

Non-orthogonal multiple access (NOMA) has received significant attention recently for its
superior spectral efficiency \cite{SystemLevelNOMA}. It is a promising candidate for
mobile communication networks, and has been included in LTE Release 13 for the scenario
of two-user downlink transmission under the name of multi-user superposition transmission
\cite{3gppNOMA}. The key idea of NOMA is to multiplex multiple users with
superposition coding at different power levels, and utilize successive interference
cancellation (SIC) at receivers with better channel conditions. Specifically, for NOMA with
two receivers, the messages to be sent are superposed with different power allocation
coefficients at the BS side. At the receivers' side, the weaker receiver decodes its intended message by treating the other's as noise, while
the stronger receiver first decodes the message of the
weaker receiver, and then decodes its own by removing the other message from the
received signal. In this way, the weaker receiver benefits from larger power, and the
stronger receiver is able to decode its own message with no interference. Hence, the overall
performance of NOMA is enhanced, compared with traditional orthogonal multiple access
schemes. It is shown in \cite{NOMA_Info} that the rate region of NOMA is the same as
the capacity region of Gaussian broadcast channels with two receivers, but with an
additional constraint that the stronger receiver is assigned less power than the weaker one.

There has been a lot of work on NOMA. In \cite{SystemLevelNOMA} and
\cite{NOMA_Info}, the authors evaluated the benefits of downlink NOMA from the
system and information theoretic perspectives, respectively. NOMA with multiple antennas
was studied in \cite{MIMONOMA}. A lot of effort has been put into the power allocation
design in NOMA. For example, the authors in \cite{NOMA_Fair} and \cite{PA_NOMA}
analyzed the necessary conditions for NOMA with two users to beat the performance of
time-division-multiple-access (TDMA), and derived closed-form expressions for the
expected data rates and outage probabilities. In \cite{JChoi_NOMA_PA}, power
allocation based on proportional fairness scheduling was investigated for downlink
NOMA. Transmit power minimization subject to rate constraints was discussed in
\cite{TransmitPowerMinimizationNOMA}.

However, all the mentioned work on NOMA has assumed a perfect knowledge of the
distributed channel state information (CSI) at the BS and all the geographically-distributed
receivers, which is difficult to realize in practice. Therefore, we consider the limited
feedback scenario wherein each receiver only has access to its own local CSI from the BS
to itself, and then broadcasts its feedback information to the BS and other receivers
\cite{Erdem_IT_Distributed_Beamforming_Relay_Interference,Cooperative_Quantization}.
Under such settings, interesting problems arise, for example: How to design simple but
efficient quantizers for NOMA? What are the performance losses compared with the
full-CSI case? A user-selection scheme based on limited feedback was studied in \cite{noma_feedback}.
 In \cite{MassiveMIMOFeedback}, the authors proposed a one-bit feedback
scheme for ordering users in downlink Massive-MIMO-NOMA systems, and derived the
achieved outage probability. In \cite{OneBitFeedback}, the authors derived the outage
probability of NOMA based on one-bit feedback of channel quality from each receiver,
and performed power allocation to minimize the outage probability. Additionally, the
problems of transmit power minimization and user fairness maximization based on
statistical CSI subject to outage constraints were studied in
\cite{NOMAOutageConstraint}. In \cite{ImperfectCSINOMA}, the authors derived the
outage probability and sum rate with fixed power allocation by assuming imperfect and
statistical CSI.

In this paper, we focus on the limited feedback design for the typical scenario of downlink
NOMA, where a BS communicates with two receivers simultaneously \cite{3gppNOMA}.
Based on distributed feedback and in the interest of user fairness, we wish to have the
minimum rate of the receivers be as large as possible. To dynamically adjust the
transmission rates for better channel utilization, we propose a uniform quantizer which
assigns each value to its left boundary point and employs variable-length encoding (VLE).
Then, power allocation is calculated based on the channel feedback. We calculate the
transmission rates that can be supported by the current channel states, and analyze the
rate loss compared with the full-CSI scenario. The derived upper bound on rate loss shows that
it decreases at least exponentially with the minimum of the feedback rates. For the
constant-rate service where outage probability is the main concern, we conversely propose
a uniform quantizer which quantizes each value to its right boundary point. Through the
developed upper bound, we show the outage probability loss also decays at least
exponentially with the minimum of feedback rate. Additionally, we analyze the achieved
diversity gain and provide a sufficient condition on the proposed quantizer in order to
achieve the full-CSI diversity order. For the general scenario with $K$ receivers, we
solve the minimum rate maximization problem within an accuracy
of $\epsilon$ in time complexity of $O\left(K\log\frac{1}{\epsilon}\right)$, and apply the
previously proposed quantizers for the two-user case here by treating the quantized channels as the perfect ones. We perform
Monte Carlo numerical simulations to verify the superiority of our proposed quantizers and
the accuracy of the theoretical analysis.

The primary goal of this paper is to study the impacts of quantization on the performance of
NOMA, and provide meaningful insights for practical limited feedback design. To
summarize, the main contributions of this paper are three-fold:
\begin{enumerate}
  \item[(1)] We propose efficient quantizers to maximize the minimum rate in NOMA.
      The ideas of our proposed quantizers and VLE as well as the designs for rate
      adaptation and outage probability based on distributed feedback can be generalized
      to many other scenarios, e.g., NOMA with other performance measures, the more
      general interference channels, and so on.
  \item[(2)] Our theoretical analysis serves as a general framework to analyze the
      performances of such quantizers in NOMA and other scenarios. For instance, it can
      be easily applied to study the performances of other power allocation schemes in
      NOMA based on limited feedback, i.e., \cite{NOMA_Fair, PA_NOMA}.
  \item[(3)] We solve the minimum rate maximization problem for any number of receivers with linear time complexity.
\end{enumerate}

The remainder of this paper is organized as follows: In Section \Rmnum{2}, we provide a
brief description of the system model and formulate the problem of limited feedback. In
Sections \Rmnum{3} and \Rmnum{4}, we propose efficient quantizers for rate adaptation
and outage probability, and analyze the performance loss. We extend our proposed
quantizers to the general case with any number of receivers in Section \Rmnum{5}.
Numerical simulations are provided in Section \Rmnum{6}. We draw the main conclusions
and summarize future work in Section \Rmnum{7}. Technical proofs are presented in the
appendices.

\textit{\textbf{Notations:}} The sets of real and natural numbers are represented by
$\mathpzc{R}$ and $\mathpzc{N}$, respectively. For any $x \in \mathpzc{R}$,
$\floor{x}$ is the largest integer that is less than or equal to x, and $\ceil{x}$ is the smallest
integer that is larger than or equal to $x$. $\text{Pr}\{\cdot\}$ and $\mathtt{E}[\cdot]$
represent the probability and expectation, respectively. For a random variable (r.v.) $X$,
$f_X(\cdot)$ is its probability density function (p.d.f.).  $\mathbbmss{CN}(\mu,
\lambda)$ represents a circularly symmetric complex Gaussian r.v. with mean $\mu$ and
variance $\lambda$. For a logical statement $\mathtt{ST}$, we let ${\pmb 1}_{\rm
\mathtt{ST}} = 1$ when $\mathtt{ST}$ is true, and ${\pmb 1}_{\rm \mathtt{ST}} = 0$
otherwise. Finally, the expression $X\sim_Y Z$ means $0 < \lim_{Y\rightarrow
\infty}\frac{X}{Z} < \infty$.

\section{Problem Formulation}

\subsection{System Model}
\begin{figure}
\centering
  \includegraphics[width=3 in]{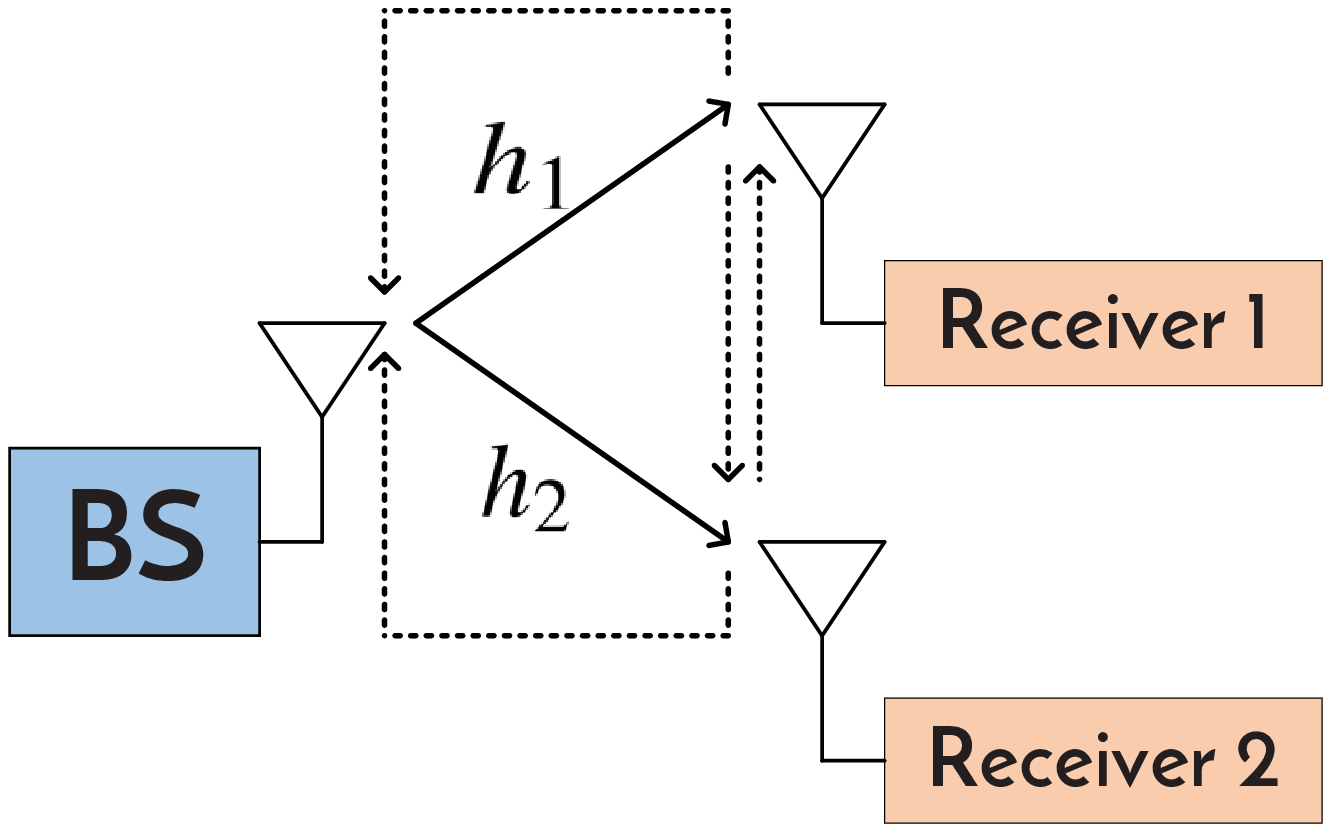}
  \caption{Downlink NOMA networks. The solid and dashed lines represent the signal and feedback links, respectively.}
\end{figure}

Consider the downlink transmission in Fig. 1, where a BS is to transmit a superposition of
two symbols to two receivers over the same resource block. Both BS and receivers are
equipped with only a single antenna. According to the multiuser superposition transmission
scheme \cite{3gppNOMA}, the transmitted signal is formed as
\begin{align}
  x = \sqrt{P_1}s_1 + \sqrt{P_2}s_2, \nonumber
\end{align}
where $s_i$ is the information bearing symbol for Receiver $i$ with
$\mathtt{E}\left[s_i\right] = 0$ and $\mathtt{E}\left[\left|s_i\right|^2\right] = 1$ for
each channel state (the expectation is over all transmitted symbols); $P_i$ is the average
transmit power associated with $s_i$. Let $P = P_1 + P_2$ be the total transmit power,
and $\alpha = \frac{P_1}{P}$ be the power allocation coefficient, then, $P_1 = \alpha P$
and $P_2 = (1-\alpha)P$ with $0 \leq \alpha \leq 1$.

Denote by $h_i \sim \mathbbmss{CN}(0, \lambda_i)$ the channel coefficient from the BS to
Receiver $i$. Without loss of generality, assume $\lambda_1 \geq \lambda_2$. The
received signals at Receivers 1 and 2 are respectively given by
\begin{eqnarray}
  y_1 = h_1\sqrt{P_1}s_1 + h_1\sqrt{P_2}s_2 + n_1, \quad y_2 = h_2\sqrt{P_1}s_1 + h_2\sqrt{P_2}s_2 + n_2, \nonumber
  \end{eqnarray}
where $n_i \sim \mathbbmss{CN}(0, 1)$ represents the background noise. Let $H_i =
\left|h_i\right|^2$, then, the p.d.f. of $H_i$ is $f_{H_i}(x) =
\frac{e^{-\frac{\lambda_i}{x}}}{\lambda_i}$ for $x > 0$.\footnote{The results in this
paper can be trivially generalized to other distributions of $H_1$ and $H_2$. } We
assume a quasi-static channel model, in which the channels vary independently from one
block to another, while remaining constant within each block. Either receiver is assumed to
perfectly estimate its local CSI (i.e., $H_i$), and send the associated quantized local CSI to
the other receiver and the BS in a broadcast manner via error-free and delay-free feedback
links \cite{David_Love_Grassmanian,Xiaoyi_Multicast_WCOM}.

With SIC, the stronger receiver with better channel condition (i.e., larger $H_i$) first
decodes the message for the weaker receiver, and then decodes its own after removing the
message of the weaker one from its received signal; the weaker receiver with poorer
channel condition directly decodes its own message by treating the message of the stronger
one as noise \cite{NOMA_Power_Allocation,NOMA_MaxMin_Power_Allocation}.
Specifically, when $H_1 \geq H_2$, the rate for Receiver 2 (i.e., the weaker one) to
decode $s_2$ by treating $s_1$ as noise is
\begin{align}
  r_2(\alpha) = \log_2\left(1 + \frac{PH_2(1-\alpha)}{\alpha H_2P + 1}\right), \nonumber
\end{align}
which is not larger than the rate for Receiver 1 to decode $s_2$, given as $r_{1\rightarrow
2} = \log_2\left(1 + \frac{PH_1(1-\alpha)}{\alpha H_1 P + 1}\right)$. If $s_2$ is
transmitted at the rate of $r_2(\alpha)$, Receiver 1 can decode $s_2$ successfully with an
arbitrarily small probability of error \cite{Elements_IT}. Afterwards, Receiver 1 can remove
$h_1\sqrt{P_2}s_2$ from $y_1$, and achieve a data rate for $s_1$ as
\begin{align}
  r_1(\alpha) & = \log_2\left(1 + \alpha P H_1\right). \nonumber
\end{align}
On the other hand, when $H_1 < H_2$, Receiver 2 first decodes $s_1$, removes
$h_2\sqrt{P_1}s_1$ from $y_2$, and then decodes $s_2$, while Receiver 1 decodes $s_1$
directly by treating $s_2$ as noise.

\subsection{Maximum Minimum Rate}

Our goal is to maximize the minimum of $r_1(\alpha)$ and $r_2(\alpha)$ to ensure
fairness between receivers \cite{Cooperative_Quantization,MaxMinFairness}. When
perfect CSI is available at the BS and receivers, the optimal power allocation coefficient
$\alpha^{\star}$ can be found by solving the optimization problem $r_{\max} = \max\limits_{0 \leq
\alpha \leq 1} \min \{r_1(\alpha), r_2(\alpha)\}$, the solution of which is given in the
following theorem.
\begin{theorem}
  When $H_1 \geq H_2$, the solution of $\max\limits_{0 \leq \alpha \leq 1} \min \{r_1(\alpha), r_2(\alpha)\}$ is given by
  \begin{align}
  \label{opt_alpha}
    \alpha^{\star}
%    = \frac{\sqrt{\left(H_1 + H_2\right)^2 + 4 H_1 H_2^2 P} - \left(H_1 + H_2\right)}{2 H_1 H_2 P}
    =
    \frac{2H_2}{\sqrt{\left(H_1 + H_2\right)^2 + 4 H_1 H_2^2 P} + \left(H_1 + H_2\right)}.
  \end{align}
\end{theorem}
\begin{proof}
Notice that with $\alpha$ increasing from $0$ to $1$, $r_1(\alpha)$ increases from 0 to
$\log_2\left(1 + PH_1\right)$ and $r_2(\alpha)$ decreases from $\log_2\left(1 +
PH_2\right)$ to $0$. Since $\log_2\left(1 + PH_1\right) \geq \log_2\left(1 +
PH_2\right)$, the maximum minimum rate is reached when $r_1(\alpha^{\star}) =
r_2(\alpha^{\star})$, from which $\alpha^{\star}$ in \eqref{opt_alpha} is derived.
\end{proof}

The expression of $\alpha^{\star}$ when $H_1 < H_2$ can be obtained straightforwardly. It is found from \eqref{opt_alpha} that: (\rmnum{1}) Both messages attain the same rate at
optimality, i.e., $r_1\left(\alpha^{\star}\right) = r_2\left(\alpha^{\star}\right) = r_{\max}$. Moreover,
it can be verified that the rate pair $(r_1\left(\alpha^{\star}\right),
r_2\left(\alpha^{\star}\right))$ is on the rate region boundaries of both NOMA and
Gaussian broadcast channels with two receivers \cite{NOMA_Info}. (\rmnum{2}) When
$P\rightarrow 0$, $\alpha^{\star}\rightarrow \frac{H_2}{H_1+H_2}$, in which case the
power assigned to the stronger receiver is in proportion to the channel quality of the weaker
one; when $P\rightarrow \infty$, $\alpha^{\star}\rightarrow 0$, then, BS should allocate
almost all the power to the weaker one. (\rmnum{3}) $\alpha^{\star} \geq \frac{1}{2}$.
Generally, NOMA steers more power towards the weaker receiver to balance their
transmissions.

It is also worth pointing out that $\alpha^{\star}$ in \eqref{opt_alpha} satisfies the
requirement for power allocation considered in \cite{PA_NOMA} and
\cite{NOMA_Fair}: the achieved individual rate should exceed that in the TDMA scheme,
i.e., $r_i(\alpha^{\star}) \geq \frac{1}{2}\log_2(1 + PH_i)$ for $i = 1, 2$. Therefore, the
maximum minimum rate we consider in this paper achieves higher rates in addition to
better fairness between receivers.

With perfect CSI, the decoding order is determined based on whether $H_1 \geq H_2$
holds. The maximum minimum rate is
\begin{align}
\label{rmax}
  r_{\max} & =
\left\{
\begin{matrix}
  \log_2\left(1 + \frac{2H_1H_2P}{\sqrt{\left(H_1 + H_2\right)^2 + 4 H_1 H_2^2 P} + \left(H_1 +
H_2\right)}\right), & H_1 \geq H_2, \\
\log_2\left(1 + \frac{2H_1H_2P}{\sqrt{\left(H_1 + H_2\right)^2 + 4 H_1^2 H_2 P} + \left(H_1 +
H_2\right)}\right), & H_1 < H_2, \\
\end{matrix}
\right.
\end{align}
and the outage probability of minimum rate is
\begin{align}
\label{out_min}
  \mathtt{out}_{\min} =
\text{Pr}\left\{r_{\max} < r_{\rm th}\right\},
\end{align}
where $r_{\rm th}$ is the data rate at which the BS will transmit $s_1$ and $s_2$ for
every channel state.

\subsection{Limited Feedback}

In the limited-feedback scenario, for an arbitrary quantizer $q:
\mathpzc{R}\rightarrow\mathpzc{R}$, Receiver $i$ maps $H_i$ to $q\left(H_i\right)$,
and feeds the index of $q\left(H_i\right)$ back to the BS and the other receiver, as shown
in Fig.1. The index of $q\left(H_i\right)$ is decoded and the value of $q\left(H_i\right)$ is
recovered. The decoding order will be contingent on whether $q\left(H_1\right) \geq
q\left(H_2\right)$. For instance, when $q(H_1) \geq q(H_2)$, Receiver 1 is considered
``stronger'', while Receiver 2 is ``weaker''. In this case, the power allocation coefficient is computed
based on \eqref{opt_alpha} by treating $q\left(H_i\right)$ as $H_i$, i.e., $\alpha_{q} =
\frac{2q(H_2)}{\sqrt{\left(q(H_1) + q(H_2)\right)^2 + 4 q(H_1)q^2(H_2)P} + q(H_1) +
q(H_2)}$.

For rate adaptation, we shall design appropriate rates $r_{1, q}$ and $r_{2, q}$ for the
messages $s_1$ and $s_2$ based on limited feedback from the two receivers, such that
$r_{1, q}$ and $r_{2, q}$ can be supported and NOMA can be performed. The
corresponding rate loss will be
\begin{align}
  r_{\rm loss} = r_{\max} - \min\left\{r_{1, q}, r_{2, q}\right\}, \nonumber
\end{align}
where $r_{\max}$ is given in \eqref{rmax}.

For a constant-rate service, we care more about whether the current channels are strong
enough to support target data rate with the power allocation coefficient computed based
on limited feedback. The achieved outage probability is $\mathtt{out}_q =
\text{Pr}\left\{r_{q} < r_{\rm th}\right\}$, where
\begin{align}
r_{q} & = \min\left\{r_1\left(\alpha_{q}\right), r_2\left(\alpha_{q}\right)\right\} \nonumber\\
& =
\left\{
\begin{matrix}
  \min\left\{\log_2\left(1 + P \times \alpha_{q} \times H_1\right), \log_2\left(1 + \frac{PH_2\left(1 - \alpha_q\right)}
  {PH_2\alpha_q + 1}\right)\right\}, & q(H_1) \geq q(H_2), \\
\min\left\{\log_2\left(1 + \frac{PH_1\left(1 - \alpha_q\right)}
  {PH_1\alpha_q + 1}\right), \log_2\left(1 + P \times \alpha_{q} \times H_2\right)\right\}, & q(H_1) < q(H_2),
\end{matrix}
\right.\nonumber
\end{align}
The outage probability loss is given as
\begin{eqnarray}
\label{loss}
\mathtt{out}_{{\rm loss}, q}  = \mathtt{out}_{\min} - \mathtt{out}_q.
\end{eqnarray}
%In the remainder of this paper, unless specified, we focus on the scenarioes of $H_1 \geq
%H_2$ in the full-CSI case and $q(H_1) \geq q(H_2)$ under limited feedback. The results
%for $H_1 < H_2$ or $q(H_1) \geq q(H_2)$ can be obtained readily by swapping the
%indices.

In the subsequent sections, we will propose efficient quantizers and investigate the
performance losses brought by limited feedback.

\section{Limited Feedback for Minimum Rate}
In this section, we first describe the proposed quantizer when the minimum rate is
the concern, then, we show the relationship between the rate loss and the feedback rates.

\subsection{Proposed Quantizer}

We consider a uniform quantizer $q_{r}: \mathpzc{R}\rightarrow\mathpzc{R}$, given
by\footnote{In $q_r$, ``$q$'' stands for quantizer, and the subscript ``$r$'' represents rate.}
\begin{align}
% \label{quantizerB}
  q_{r}({x}) = \left\{
  \begin{matrix}
    \left\lfloor \frac{x}{\Delta}\right\rfloor\times \Delta , & x \leq T\Delta,\\
    T\Delta, & x > T\Delta,
  \end{matrix}
  \right. \nonumber
\end{align}
where the bin size $\Delta$ and the maximum number of bins $T \in \mathpzc{N}$ are
adjustable parameters. As shown in Fig. 2, $q_r(x)$ quantizes $x$ to the left boundary of
the interval where $x$ is. For any $x\in [n\Delta, (n+1)\Delta)$ when $0 \leq n \leq T-1$,
we have $q_{r}({x}) = n\Delta$ and $ x - \Delta \leq q_r(x) \leq x $; for any $x\in
[T\Delta, \infty)$, $q_{r}({x}) = T\Delta$ and $q_r(x) \leq x$.

\begin{figure}
\centering
  \includegraphics[width=3in]{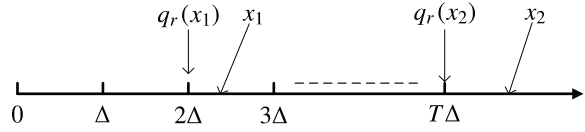}
  \caption{A uniform quantizer for minimum rate.}
\end{figure}

%Rather than imposing a constraint $T\Delta$ on $x$, we can make a relaxation and
%consider a unbounded uniform quantizer depicted in Fig. 3, given by
%\begin{align}
%%\label{quantizerUB}
%  q_{r}({x}) = \left\lfloor \frac{x}{\Delta}\right\rfloor\times \Delta, \nonumber
%\end{align}
%For any $x\in [n\Delta, (n+1)\Delta)$, $q_{r}({x}) = n\Delta$ and $ x - \Delta \leq
%q_{r}({x}) \leq x$.
%
%\begin{figure}[H]
%\centering
%  \includegraphics[width=3.5in]{q_r_bar.eps}
%  \caption{Unbounded uniform quantizer for minimum rate.}
%\end{figure}

\subsection{Rate Adaptation and Loss}

When $q_{r}\left(\cdot\right)$ is employed, Receiver 2 is viewed as the ``weak'' receiver if
$q_{r}\left(H_1\right) \geq q_{r}\left(H_2\right)$. Then, according to \eqref{opt_alpha}, the
power allocation coefficient $\alpha_{q_r}$ is calculated as
\begin{align}
  {\alpha}_{q_{r}} =
  \left\{
  \begin{matrix}
  \frac{2 q_{r}\left(H_2\right)}{\sqrt{\left[q_{r}\left(H_1\right) + q_{r}\left(H_2\right)\right]^2
+ 4 q_{r}\left(H_1\right) q_{r}^2\left(H_2\right) P}
   + \left[q_{r}\left(H_1\right) + q_{r}\left(H_2\right)\right]},
   & q_{r}\left(H_1\right) > 0, q_{r}\left(H_2\right) > 0,\\
   0, & q_{r}\left(H_1\right) = 0 \text{ or } q_{r}\left(H_2\right) = 0,
  \end{matrix}
  \right.\nonumber
\end{align}
which satisfies $\log_2\left(1 + P \times {\alpha}_{q_{r}} \times q_{r}\left(H_1\right)\right)
= \log_2\left(1 + \frac{q_{r}\left(H_2\right)\times
\left(1-{\alpha}_{q_{r}}\right)}{{\alpha}_{q_{r}} \times q_{r}\left(H_2\right) +
\frac{1}{P}}\right)$ when ${\alpha}_{q_{r}} \neq 0$. To exploit the channels as much as
possible, we let the BS send messages $s_1$ and $s_2$ at rates of
\begin{align}
\label{achievedRate}
  {r}_{1, q_{r}}  = \log_2\left(1 + P \times {\alpha}_{q_{r}} \times q_{r}\left(H_1\right) \right),
  {r}_{2, q_{r}}  = \log_2\left(1 + \frac{P\times q_r(H_2)\left(1 - \alpha_{q_r}\right)}{P\times q_r(H_2)\alpha_{q_r} + 1}\right).
\end{align}

\begin{lemma}
  When $q_r\left(H_1\right) \geq q_r\left(H_2\right)$, the rates ${r}_{1, q_{r}}$ and ${r}_{2, q_{r}}$ in \eqref{achievedRate} can be achieved.
\end{lemma}
\begin{proof}
Based on the channel coding theorem \cite{Elements_IT}, if we can show the channel
capacities for $s_1$ and $s_2$ under the settings of NOMA are no smaller than ${r}_{1,
q_{r}}$ and ${r}_{2, q_{r}}$, the rates ${r}_{1, q_{r}}$ and ${r}_{2, q_{r}}$ can be
achieved with a probability of error that can be made arbitrarily small.

When $q_{r}\left(H_1\right) = 0 \text{ or } q_{r}\left(H_2\right) = 0$, it is trivial to verify that ${r}_{1, q_{r}}$ and
${r}_{2, q_{r}} $ can be supported. When $q_r\left(H_1\right) \geq
q_r\left(H_2\right) > 0$, the channel capacity for Receiver 2 by treating $s_1$ as noise is
$r_2 = \log_2\left(1 + \frac{H_2(1-{\alpha}_{q_{r}})}{{\alpha}_{q_{r}} \times H_2 +
\frac{1}{P}}\right) \geq \log_2\left(1 + \frac{q_r\left(H_2\right)\times
(1-{\alpha}_{q_{r}})}{{\alpha}_{q_{r}} \times q_r\left(H_2\right)  + \frac{1}{P}}\right) =
r_{2, q_r}$, since $\log_2\left(1 + \frac{x(1 - \alpha)}{x \alpha +\frac{1}{P}}\right)$ is an
increasing function of $x$ and $q_r(H_2) \leq H_2$. At the side of Receiver 1, the channel
capacity of $s_2$ with treating $s_1$ as noise is $r_{1\rightarrow 2} = \log_2\left(1 +
\frac{H_1(1-{\alpha}_{q_{r}})}{{\alpha}_{q_{r}} \times H_1  + \frac{1}{P}}\right) \geq
\log_2\left(1 + \frac{q_r\left(H_1\right)\times (1-{\alpha}_{q_{r}})}{{\alpha}_{q_{r}} \times
q_r\left(H_1\right)  + \frac{1}{P}}\right) \geq \log_2\left(1 +
\frac{q_r\left(H_2\right)\times (1-{\alpha}_{q_{r}})}{{\alpha}_{q_{r}} \times
q_r\left(H_2\right)  + \frac{1}{P}}\right) = r_{2, q_r}$, because $H_1 \geq q_r(H_1) \geq
q_r(H_2)$. Hence, $s_2$ can be decoded at Receiver 1 with an arbitrarily small error and
removed from $y_1$. After that, the channel capacity of $s_1$ is $r_1 = \log_2\left(1 +
P\times {\alpha}_{q_{r}}\times H_1\right) \geq \log_2\left(1 + P \times
{\alpha}_{q_{r}}\times q_r\left(H_1\right)\right) = r_{1, q_r}$. Therefore, the rates $r_{1,
q_r}$ and $r_{2, q_r}$ can be achieved for both $s_1$ and $s_2$.
\end{proof}

To sum up, it is the key fact of $q_r(x) \geq x$ that ensures the rates ${r}_{1, q_{r}}$ and
${r}_{2, q_{r}}$ in \eqref{achievedRate} can be supported. When $q_r(H_1) \geq
q_r(H_2)$, the rate loss is defined as
\begin{align}
%\label{rloss}
  r_{{\rm loss}} %& = \min \left\{r_1\left(\alpha^{\star}\right), r_2\left(\alpha^{\star}\right)\right\} - \min \{r_{1, q_r}, r_{2, q_r}\} \nonumber\\
%  %  & = \min\left\{\log_2\left(1 + \alpha^{\star} P H_1\right), \log_2\left(1 + \frac{PH_2(1-\alpha^{\star})}{\alpha^{\star} H_2P + 1}\right)\right\}
%%  - \min\left\{\log_2\left(1 + P \times {\alpha}_{q_{r}} \times q_r\left(H_1\right)\right), \log_2\left(1 + \frac{q_r\left(H_2\right)\times
%%  (1-{\alpha}_{q_{r}})}{{\alpha}_{q_{r}} \times q_r\left(H_2\right) + \frac{1}{P}}\right)\right\}\nonumber\\
   & =
   r_{\max} - \min \{r_{1, q_r}, r_{2, q_r}\}. \nonumber
\end{align}

%
%Similarly, for the unbounded uniform quantizer $q_{r}(\cdot)$, when
%$q_{r}(H_1) \geq q_{r}(H_2)$, the rates $r_{1, q_{r}} = \log_2\left(1 +
%P\times \alpha_{q_{r}}\times q_{r}\left(H_1\right)\right) = r_{2, q_{r}}
%= \log_2\left(1 + \frac{q_{r}\left(H_2\right)
%(1-\alpha_{q_{r}})}{\alpha_{q_{r}} \times q_{r}\left(H_2\right) +
%\frac{1}{P}}\right)$ can be achieved for both messages, and the corresponding rate loss is
%\begin{align}
%  r_{{\rm loss}} & =   \min \left\{r_1\left(\alpha^{\star}\right), r_2\left(\alpha^{\star}\right)\right\} -
%  \min \left\{r_{1, q_r}, r_{2, q_r}\right\} \nonumber\\
%   & = \log_2\left(1 + \alpha^{\star} P H_1\right) - \log_2\left(1 + P \times \alpha_{q_{r}} \times q_{r}\left(H_1\right)\right). \nonumber
%\end{align}

\begin{lemma}
  The average rate loss of the quantizer $q_{r}(\cdot)$ is upper-bounded by:
  \begin{align}
  \label{rateBound}
    & \mathtt{E}\left[r_{{\rm loss}}\right] \leq \log_2\left(1 + C_0 \times P \times
    \max\left\{e^{-\frac{T\Delta}{\lambda_1}}, \Delta\right\}\right),
%     \label{rateunBound}
%    & \mathtt{E}\left[r_{{\rm loss}}\right] \leq \log_2\left(1 + C_1 \times P \times \Delta\right),
  \end{align}
  where $C_0$ is a positive constant that is independent of $P, T$ and $\Delta$.
\end{lemma}
\begin{proof}
  See Appendix A.
\end{proof}

We mainly focus on showing how the average rate loss changes with the bin size $\Delta$.
It is beyond the scope of this paper to find the tightest bounds, i.e., the smallest value for
$C_0$. A value for $C_0$ which is derived from the proof in Appendix A is $C_0 =
\max\left\{4+\frac{\lambda_1}{\lambda_2}, \lambda_2\right\}$.

It is observed from \eqref{rateBound} that when
$e^{-\frac{T\Delta}{\lambda_1}} > \Delta$, the maximum number of bins, $T$, can degrade the rate. To eliminate this effect, we choose $T$ such that
$e^{-\frac{T\Delta}{\lambda_1}} = \Delta$, which yields $T =
\frac{\lambda_1}{\Delta}\log\frac{1}{\Delta}$.\footnote{Approaching the performance in
the full-CSI case generally requires a small value for $\Delta$. We mainly consider the
case where $\Delta \leq 1$ in this paper.} With an appropriate value for $T$, we can
make the rate loss decrease at least linearly with respect to $\Delta$.

\begin{corollary}
  When $T=\frac{\lambda_1}{\Delta}\log\frac{1}{\Delta}$, the average rate loss of the quantizer
$q_{r}(\cdot)$ is upper-bounded by:
  \begin{align}
  \label{rateBound_Linear}
  \mathtt{E}\left[r_{{\rm loss}}\right] \leq \log_2\left(1 + C_0 \times P \times \Delta\right) \leq C_1 \times P \times \Delta,
%     \label{rateunBound}
%    & \mathtt{E}\left[r_{{\rm loss}}\right] \leq \log_2\left(1 + C_1 \times P \times \Delta\right),
  \end{align}
  where $C_0$ and $C_1$ are positive constants that are independent of $P$ and $\Delta$.
\end{corollary}

\subsection{Feedback Rate}

Rather than the naive fixed-length encoding (FLE) for feedback information which requires
$\lceil\log_2(T+1)\rceil$ bits per receiver per channel state, we consider the more efficient
variable-length encoding (VLE)
\cite{Erdem_VLQ_IT,Xiaoyi_Multicast_WCOM}.\footnote{For example, when $\Delta
= 0.01$ and $\lambda_1 = 1$, $T = \frac{\lambda_1}{\Delta}\log\frac{1}{\Delta}
\approx 460.5$. When FLE is adopted, the feedback rate per receiver will be
$\lceil\log_2(T+1)\rceil$ = 9 bits per channel state. As shown by the theoretical analysis
and numerical simulations later, VLE will cost far fewer bits.} An example of VLE that can
be applied here is $b_{0} = \{0\}$, $b_1 = \{1\}$, $b_2 = \{00\}$, $b_3 = \{01\}$ and so on,
sequentially for all codewords in the set $\{0, 1, 00, 01, 10, 11, \ldots\}$, where
$b_n$ is the binary string to be fed back when $q_{r}(x) = n\Delta$. The length of $b_n$
is $\floor{\log_2(n+2)}$. The following theorem derives an upper bound on the rate loss with respect to the feedback rate of Receiver $i$ (denoted by $R_{r, {\rm VLE}, i}$).

\begin{theorem}
When variable-length encoding is applied to the quantizer $q_{r}(\cdot)$, the rate loss
decays at least exponentially with the number of feedback bits:
\begin{align}
\label{rate_VLE}
\mathtt{E}\left[r_{{\rm loss}}\right] \leq \log_2\left(1 + C_2 \times P \times
2^{-\min\left\{R_{r, {\rm VLE}, 1}, R_{r, {\rm VLE}, 2}\right\}}\right) \leq C_3 \times P \times
2^{-\min\left\{R_{r, {\rm VLE}, 1}, R_{r, {\rm VLE}, 2}\right\}},
\end{align}
where $C_2$ and $C_3$ are positive constants independent of $P$ and
$R_{r, {\rm VLE}, i}$.
\end{theorem}
\begin{proof}
The feedback rate of Receiver $i$ is derived as
\begin{align}
R_{r, {\rm VLE}, i}
& = \sum_{n = 0}^{T-1} \floor{\log_2(n+2)} \int_{n\Delta}^{(n+1)\Delta} f_{H_i}(H_i) {\rm d}H_i
+ \floor{\log_2(T+2)} \int_{T\Delta}^{\infty} f_{H_i}(H_i) {\rm d}H_i \nonumber\\
& \leq \sum_{n = 0}^{\infty} \floor{\log_2(n+2)} \int_{n\Delta}^{(n+1)\Delta} f_{H_i}(H_i) {\rm d}H_i \nonumber\\
& \leq \sum_{n = 0}^{\infty} \underbrace{\log_2(n+2)}_{\leq \log_2(n+1)+1} {\int_{n\Delta}^{(n+1)\Delta}
 \frac{e^{-\frac{H_i}{\lambda_i}}}{\lambda_i} {\rm d}H_i} \nonumber\\
& \leq \sum_{n = 0}^{\infty} e^{-\frac{n\Delta}{\lambda_i}} \left(1-e^{-\frac{\Delta}{\lambda_i}}\right) \times \log_2(n+1)
+ \underbrace{\sum_{n = 0}^{\infty} 1 \times \int_{n\Delta}^{(n+1)\Delta} \frac{e^{-\frac{H_i}{\lambda_i}}}{\lambda_i} {\rm d}H_i}_{=1} \nonumber\\
& = 1 + {\left(1-e^{-\frac{\Delta}{\lambda_i}}\right)}
\sum_{n = 0}^{\infty}e^{-\frac{n\Delta}{\lambda_i}} \times \log_2(n+1)
 \leq 1 + \frac{\Delta}{\lambda_i}
\sum_{n = 0}^{\infty}e^{-\frac{n\Delta}{\lambda_i}} \times \log_2(n+1).\nonumber
\end{align}
With the help of \cite[Eq.(22)]{Xiaoyi_Multicast_WCOM}: $\sum_{n = 1}^{\infty}
e^{-\beta n}\log(n) \leq \frac{e^{-\beta}}{\beta}\left[2 + \log\left(1 +
\frac{1}{\beta}\right)\right]$, by letting $\beta = e^{-\frac{\Delta}{\lambda_i}}$, we have
\begin{align}
  \sum_{n = 0}^{\infty}
e^{-\frac{n\Delta}{\lambda_i}}  \times\log_2(n+1)
 & = \sum_{n = 1}^{\infty}
e^{-\frac{n\Delta}{\lambda_i}}  \times\log_2(n+1)\nonumber\\
& = \frac{e^{\frac{\Delta}{\lambda_i}}}{\log 2} \sum_{n = 2}^{\infty}
e^{-\frac{n\Delta}{\lambda_i}}  \times\log(n)
 \leq
\frac{1}{\frac{\Delta}{\lambda_i}}\left[\frac{2}{\log 2} + \log_2\left(1 +
\frac{1}{\frac{\Delta}{\lambda_i}}\right)\right].\nonumber
\end{align}
Then,  $R_{r, {\rm VLE}, i}$ is upper-bounded by\footnote{Although it is intractable to
derive a closed-form expression for $R_{r, {\rm VLE}, i}$, the upper bound in
\eqref{fr_bound} provides a good estimate on how many feedback bits will be consumed. }
\begin{align}
\label{fr_bound}
R_{r, {\rm VLE}, i} \leq \frac{2}{\log 2} + 1  + \log_2\left(1 + \frac{1}{\frac{\Delta}{\lambda_i}}\right),
\end{align}
or equivalently (when $R_{r, {\rm VLE}, i}$ is sufficiently large),
\begin{align}
\label{deltaBoundVle}
\Delta \leq \frac{\lambda_i}{2^{R_{r, {\rm VLE}, i}-1-\frac{2}{\log 2}}-1}
\leq \frac{\lambda_i}{2^{R_{r, {\rm VLE}, i}-2-\frac{2}{\log 2}}} = C_4 \times 2^{-R_{r, {\rm VLE}, i}}.
\end{align}
Substituting \eqref{deltaBoundVle} into \eqref{rateBound_Linear} proves the theorem.
\end{proof}

%\begin{theorem}
%When variable-length encoding is applied to the quantizer $q_{r}(\cdot)$, the rate loss
%decays at least exponentially with the number of feedback bits:
%\begin{align}
%\label{rate_VLE}
%\mathtt{E}\left[r_{{\rm loss}}\right] \leq \log_2\left(1 + C_3 \times P \times
%2^{-\min\left\{R_{r, {\rm VLE}, 1}, R_{r, {\rm VLE}, 2}\right\}}\right) \leq C_4 \times P \times
%2^{-\min\left\{R_{r, {\rm VLE}, 1}, R_{r, {\rm VLE}, 2}\right\}},
%\end{align}
%where $C_3$ and $C_4$ are positive constants independent of $P$ and
%$R_{r, {\rm VLE}, i}$.
%\end{theorem}

%From \eqref{rate_VLE}, we further obtain $\mathtt{E}\left[r_{{\rm loss}}\right] \leq
%\log_2\left(1 + C_5 \times P \times 2^{-R_{r, {\rm VLE}, i}}\right) \leq  C_5 \times P
%\times C_5 \times P \times 2^{-R_{r, {\rm VLE}, i}}$.
Therefore, we can see that appropriate values for $T$ and the use of VLE enable the rate loss
to decrease at least exponentially with the feedback rate.

\section{Limited Feedback for Outage Probability}

Outage probability is an important performance metric that evaluates the chance that the
channels are not strong enough to support the constant-rate data service
\cite{spaceTimeCodingBook}. An ideal quantizer for outage probability should have at least
 the following properties: (\rmnum{1}) The outage probability loss should decrease
toward zero when the feedback rate increases toward infinity. (\rmnum{2}) The outage
probability loss should approach zero whenever $P\rightarrow 0$ or $P\rightarrow
\infty$. The intuition of (\rmnum{2}) comes from the fact that when $P$ is adequately
small, the outage probabilities of both the full-CSI case and the quantizer should be close
to one; when $P$ is significantly large, both outage probabilities should be almost zero.
Then, the outage probability losses in both scenarios go to zero.

\subsection{Proposed Quantizer}
As portrayed in Fig. \ref{out_q}, the uniform quantizer proposed for outage probability is given by
\begin{align}
\label{quantizerBOut}
  q_{o}({x}) = \left\{
  \begin{matrix}
    \left\lceil \frac{x}{\Delta}\right\rceil\times \Delta , & x \leq T\Delta,\\
    (T+1)\Delta, & x > T\Delta.
  \end{matrix}
  \right.
\end{align}

\begin{figure}
\centering
  \includegraphics[width=4in]{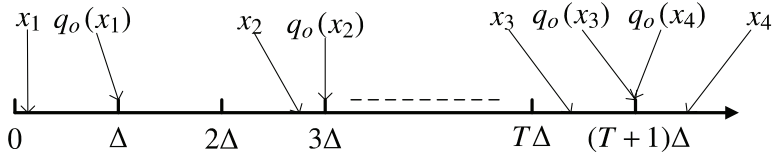}
  \caption{A uniform quantizer for outage probability.}
  \label{out_q}
\end{figure}

%\begin{figure}[H]
%\centering
%  \includegraphics[width=3in]{q_o_bar.eps}
%  \caption{Unbounded uniform quantizer for outage probability.}
%\end{figure}

The only difference between $q_o(\cdot)$ and $q_r(\cdot)$ lies in whether the left or right
boundary of the interval is used as the reconstruction point. The quantizer proposed for rate
adaptation cannot be directly inherited because when the channel is very weak (i.e., ${H}_i
< \Delta$), it will be quantized as zero (i.e., $q_r(H_i) = 0$), which will result in a
zero-value power allocation coefficient, i.e.,$\alpha_{q_r} = 0$, and a minimum rate of
zero, i.e., $r_1\left(\alpha_{q_r}\right) = 0$ or $r_2\left(\alpha_{q_r}\right)$. In this
case, the transmission will surely encounter an outage. However, even a weak channel
reserves the possibility of non-outage, so long as the transmit power $P$ is large enough.
Therefore, an appropriate quantizer for outage probability should not quantize any value
to zero. The quantizer in \eqref{quantizerBOut} fulfills this requirement.

\subsection{Outage Probability Loss}

\begin{lemma}
  The outage probability loss of the quantizer $q_{o}(\cdot)$ is upper-bounded by:
  \begin{align}
  \label{outBound}
    & \mathtt{out}_{{\rm loss}, q_{o}} \leq C_5 \times e^{-\frac{C_6}{P}}
    \times \frac{1 + \sqrt{P}}{P} \times \max\left\{\Delta^{\frac{1}{2}},
          \Delta^{\frac{3}{2}},  e^{-\frac{T\Delta}{\lambda_1}}\right\},
%    \\
%     \label{outunBound}
%    & \mathtt{out}_{{\rm loss}, q_{o}} \leq C_8 \times
%    e^{-\frac{C_9}{P}}   \times \frac{1 + \sqrt{P}}{P} \times \max\left\{\sqrt{\Delta}, \Delta\right\},
  \end{align}
  where $C_5$ and $C_6$ are positive constants that are independent of $P$ and $\Delta$.
\end{lemma}
\begin{proof}
  See Appendix B.
\end{proof}

Different from the rate loss which increases linearly in terms of $P$, because of the term
$e^{-\frac{C_6}{P}} \times \frac{1 + \sqrt{P}}{P} $, the upper bound on
$\mathtt{out}_{{\rm loss}, q_{o}}$ in \eqref{outBound} converges to zero either when
$P\rightarrow 0$ or $P\rightarrow \infty$.

To have good performance, we mainly focus on the quantizers with small granularities. When $\Delta \leq 1$, we have $\Delta^{\frac{3}{2}} \leq
\Delta^{\frac{1}{2}}$, and the upper bound in \eqref{outBound} is restricted by
$\max\left\{e^{-\frac{T\Delta}{\lambda_1}}, \Delta^{\frac{1}{2}}\right\}$. For fixed
$\Delta$, the optimal choice for $T$ should satisfy $e^{-\frac{T\Delta}{\lambda_1}} =
\Delta^{\frac{1}{2}}$, given by $T = \frac{\lambda_1}{2\Delta}\log\frac{1}{\Delta}$.
\begin{corollary}
  When $0 < \Delta \leq 1$ and $T=\frac{\lambda_1}{2\Delta}\log\frac{1}{\Delta}$, the average rate loss of the quantizer
$q_{o}(\cdot)$ is upper-bounded by:
  \begin{align}
   \label{outBound_Linear}
    & \mathtt{out}_{{\rm loss}, q_{o}} \leq C_5 \times e^{-\frac{C_6}{P}}
    \times \frac{1 + \sqrt{P}}{P} \times \Delta^{\frac{1}{2}},
  \end{align}
  where $C_5$ and $C_6$ are positive constants independent of $P$ and $\Delta$.
\end{corollary}

\subsection{Feedback Rate}

The same VLQ for rate adaptation can be applied to $q_{o}(\cdot)$ for a better
utilization of the feedback resource. From \eqref{fr_bound} and \eqref{deltaBoundVle},
we obtain $R_{o, {\rm VLE}, i} \leq \frac{2}{\log 2} + 1  + \log_2\left(1 +
\frac{1}{\frac{\Delta}{\lambda_i}}\right)$ and $\Delta \leq C_4 \times 2^{-R_{o, {\rm
VLE}, i}}$. Thus, $\Delta^{\frac{1}{2}} \leq \sqrt{ C_4 \times 2^{-R_{o, {\rm VLE}, i}}} =
C_{7} \times 2^{-\frac{R_{o, {\rm VLE}, i}}{2}} \leq C_7 \times  2^{-\frac{\min\left\{R_{o, {\rm VLE},
1}, R_{o, {\rm VLE}, 2}\right\}}{2}}$. The following theorem states the relationship
between the outage probability loss of $q_{o}(\cdot)$ and the feedback rates.

\begin{theorem}
When variable-length encoding is applied to the quantizer $q_{o}(\cdot)$, the rate loss
decays at least exponentially as:
\begin{align}
\mathtt{out}_{{\rm loss}, q_{o}} \leq
C_{8} \times
    e^{-\frac{C_{6}}{P}} \times \frac{1 + \sqrt{P}}{P}  \times 2^{-\frac{\min\left\{
    R_{o, {\rm VLE}, 1}, R_{o, {\rm VLE}, 2}
    \right\}}{2}},
\end{align}
where $C_{6}$ and $C_{8}$ are positive constants independent of $P$
and $R_{o, {\rm VLE}, i}$.
\end{theorem}

\subsection{Diversity Order}

With an outage probability $\mathtt{out}$, the achieved diversity order is given as $d =
\lim_{P\rightarrow\infty}\frac{\log \mathtt{out}}{\log
P}$\cite{spaceTimeCodingBook}. The
following lemma shows the achievable diversity order of $q_{o}(\cdot)$ and a sufficient
condition to achieve the maximum diversity order in the full-CSI scenario.

\begin{lemma}
  \begin{enumerate}
    \item[(1)] With $q_o(\cdot)$ and fixed $\Delta$, the diversity orders of $\frac{1}{2}$
        and 1 are achievable for Receivers 1 and 2, respectively.
    \item[(2)] A sufficient condition for both receivers to achieve the maximum diversity
        order of 1 is $\Delta \sim_P P^{-\frac{1}{3}}$.
  \end{enumerate}
\end{lemma}
\begin{proof}
  See Appendix C.
\end{proof}

In the full-CSI case, both receivers can achieve the same diversity order of $1$ as in the
case when no interference exists. In the limited feedback case, it can be found from the
proofs in Appendices B and C that the cause of this insufficient diversity order for Receiver
$1$ comes from the marginal region when $0 < H_1, H_2 \leq \Delta$. Therefore, an
adequately small $\Delta$ that scales at least in proportion to $P^{-\frac{1}{3}}$ in the
high-$P$ region is desired to diminish the probability that $H_i$ falls into that region so as
to obtain the maximum diversity gain.

\section{Extension to More than Two Receivers}

\subsection{Full-CSI Performance}

In this section, we consider NOMA with more than two downlink receivers.
Assuming perfect CSI universally available and $H_1 \geq H_2 \geq \cdots \geq H_K$,
the maximum minimum rate can be obtained by solving the optimization problem:
\begin{align}
\label{kuser}
  r_{\max} =& \max\limits_{{\pmb \alpha} = \left[\alpha_1, \ldots,\alpha_K\right]}
\min\limits_{k = 1, \ldots, K} r_k ({\pmb \alpha}), \text{ subject to } 0 \leq \alpha_k\leq 1, \sum_{k=1}^K \alpha_k = 1,
\end{align}
where $K$ is the number of receivers, and $r_k({\pmb \alpha}) = \log_2\left(1 +
\frac{\alpha_k}{\sum_{i=1}^{k-1} \alpha_i + \frac{1}{PH_k}}\right)$ is the achieved rate
for Receiver $k$ under superposition coding and SIC. To the best of our knowledge, no
closed-form solution for $r_{\max}$ is available in the literature. We present the following lemma that helps solving the above optimization problem numerically.

 \begin{lemma}
   There exists ${\pmb \alpha}^{\star} = \left[\alpha_1^{\star}, \alpha_2^{\star}, \ldots, \alpha_K^{\star}\right]$, such that all
      receivers achieve the same rate at optimality, i.e., $r_{\max} = r_{1}\left({\pmb
      \alpha}^{\star}\right) = r_{2}\left({\pmb \alpha}^{\star}\right) = \cdots =
      r_{K}\left({\pmb \alpha}^{\star}\right)$.
 \end{lemma}

 The proof of Lemma 5 is given in Appendix D.  Since $r_{\max} = r_{k}\left({\pmb \alpha}^{\star}\right) = \log_2\left(1 +
      \frac{\alpha_k^{\star}}{ \sum_{i=1}^{k-1} \alpha_i^{\star} + \frac{1}{PH_k}}\right)$
      for $k = 1, \ldots, K$, we have $\alpha_k^{\star} =
      \left(2^{r_{\max}}-1\right)\times\left(\sum_{i=1}^{k-1} \alpha_i^{\star} + \frac{1}{P
      H_k}\right)$, which leads to
      \begin{align}
        \alpha_k^{\star} =
      \left(2^{r_{\max}}-1\right)\left[\frac{1}{PH_k}+\left(2^{r_{\max}}-1\right)\sum_{i=1}^{k-1}
      \frac{2^{\left(k-1-i\right)r_{\max}}}{PH_i}\right].
      \end{align}
       To find $\alpha_k^{\star}$, we need to solve for $r_{\max}$ first. Summing both
      sides from $k = 1, \ldots, K$ and after trivial calculations, we obtain
      \begin{align}
        \sum_{k=1}^K
      \alpha_k^{\star} = 1 =\underbrace{\left(2^{r_{\max}}-1\right)\sum_{i=1}^K \frac{2^{\left(K-i\right)r_{\max}}}{PH_i}}
      _{= \varpi\left(r_{\max}\right)}.
      \end{align}
      In other words, $r_{\max}$ satisfies $\varpi\left(r_{\max}\right) = 1$.\footnote{ Note that \cite{excessPower} has solved a different optimization problem, i.e. maximizing the sum rate subject to a minimum rate constraint, which satisfies $\sum_{k=1}^K \alpha_k^{\star} = 1$ but results in different $\alpha_k^{\star}$s.}

      Let $r_{\rm ub} = \log_2\left(1 + \min_{k=1, \ldots, K} P H_k\right) =
          \log_2(1 + PH_K)$. Since $\varpi\left(x\right)$ is an increasing function of $x$ as well as
          $\varpi(0) < 1$ and $\varpi(r_{\rm ub}) \geq 1$, we could use the bisection method to find the root of
$\varpi(x) = 1$ in the interval $\left(0, r_{\rm ub}\right]$. The calculation
          of $\varpi\left(x\right)$ costs $O(K)$, thus, the time complexity of
          finding $r_{\max}$ within an accuracy of $\epsilon$ is $O\left(K\log\frac{1}{\epsilon}\right)$.

\subsection{Limited Feedback}

 Under limited feedback, the previously proposed quantizers $q_r\left(\cdot\right)$ and $q_o\left(\cdot\right)$ in Figs. 2 and 3
 can still be applied here for rate adaptation and outage probability, respectively. The maximum minimum rate can be calculated
using the bisection method by treating $q_r\left(H_k\right)$ or $q_o\left(H_k\right)$ as
$H_k$, and the corresponding power allocation coefficients can be
computed. Although it is non-trivial to derive upper bounds on the losses in rate or outage
 probability for $K > 2$ theoretically, numerical simulations in Section \Rmnum{6} show that the relationships between
  the performance loss and the feedback rate are similar to the case of $K = 2$.

\section{Numerical Simulations and Discussions}

In this section, we perform numerical simulations to validate the effectiveness of our
proposed quantizers for rate adaptation and outage probability. In all subsequent
simulations for $K = 2$ receivers, we assume the channel variances are $\lambda_1 = 1$ and $\lambda_2 =
0.5$. Results for other values of $\lambda_1$ and $\lambda_2$ will exhibit
similar observations. For outage probability, sufficiently large number of channel
realizations are generated to observe at least $10000$ outage events.

\begin{figure}
\centering
\begin{minipage}{.5\textwidth}
  \centering
  \includegraphics[width=.9\linewidth]{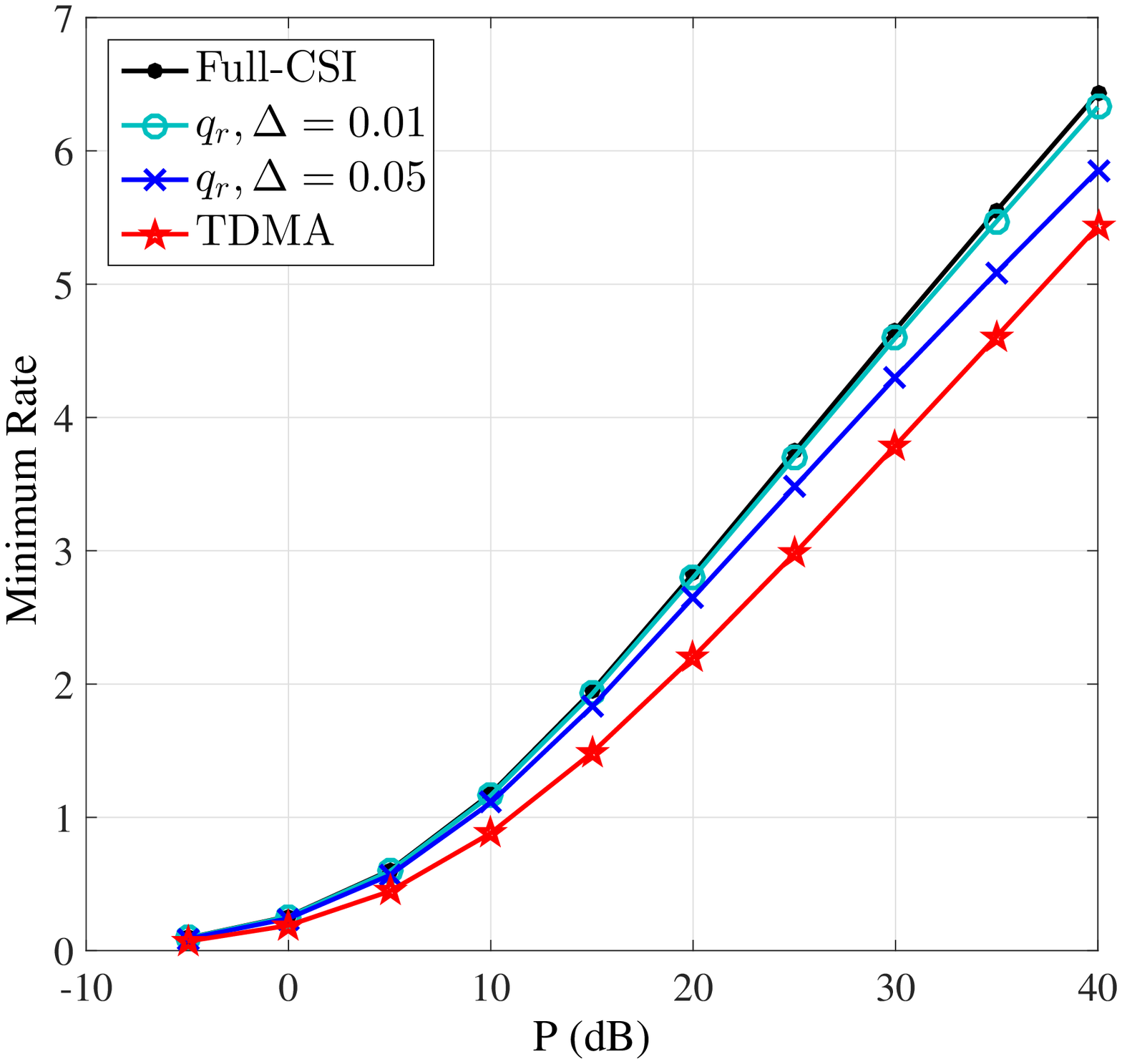}
  \captionof{figure}{Simulated minimum rates of NOMA.}
  \label{fig:test1}
\end{minipage}%
\begin{minipage}{.5\textwidth}
  \centering
  \includegraphics[width=.9\linewidth]{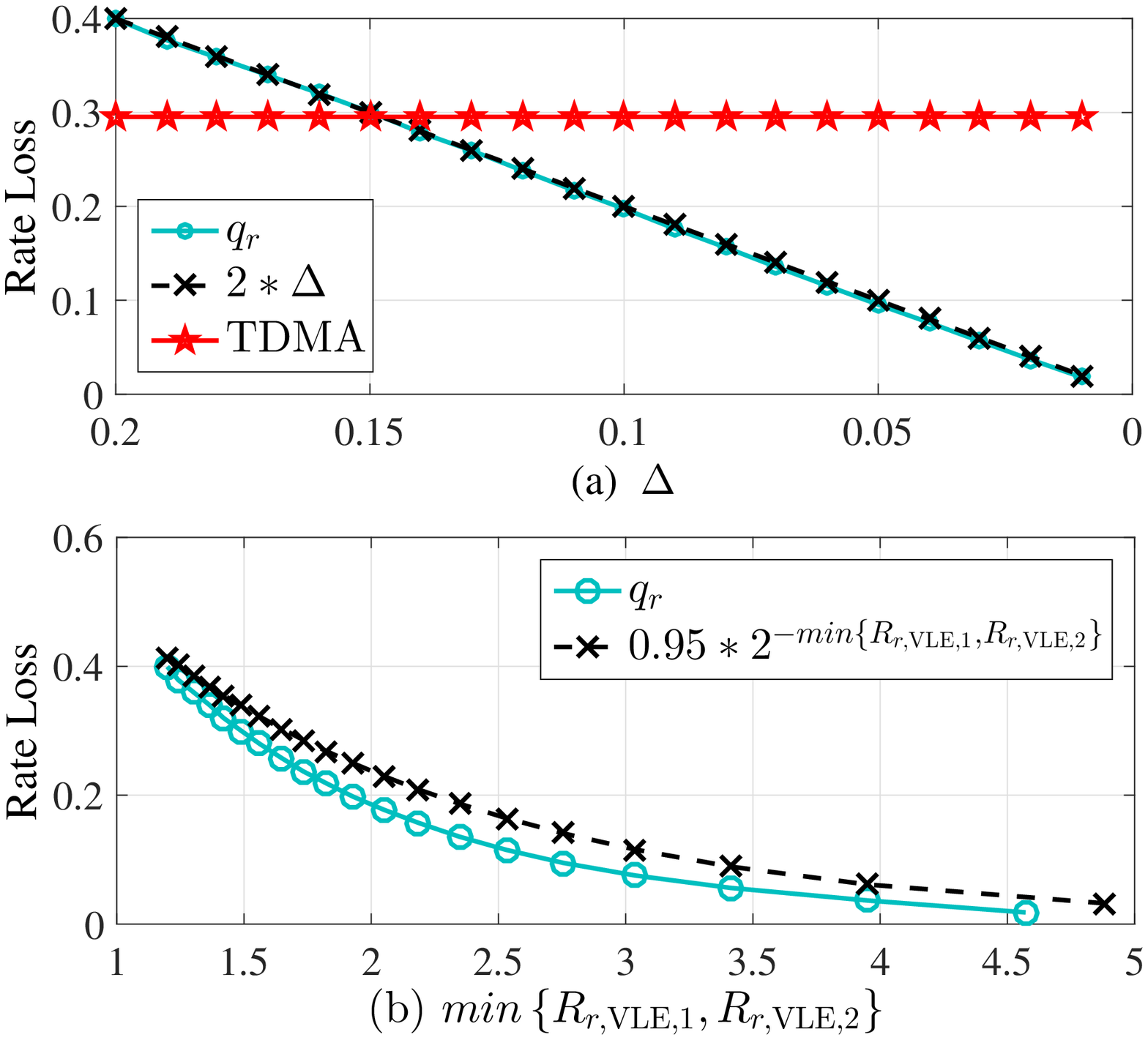}
  \captionof{figure}{Simulated rate losses versus (a) $\Delta$ and (b) $\min\left\{R_{r, {\rm VLE}, 1}, R_{r, {\rm VLE}, 2}\right\}$ for $P = 10$ dB.}
  \label{fig:test2}
\end{minipage}
\end{figure}

In Fig. \ref{fig:test1}, we simulated the minimum rates of the full-CSI case, $q_{r}(\cdot)$
and the TDMA scheme (where each receiver occupies half of the time to transmit). We
observe that the proposed quantizer with NOMA outperforms the TDMA scheme when
$\Delta = 0.01$ and $0.05$. The rate loss between the full-CSI case and $q_{r}(\cdot)$
with $\Delta = 0.01$ is almost negligible. The corresponding values for $T =
\frac{\lambda_1}{\Delta}\log\frac{1}{\Delta}$ and the feedback rates for both receivers
(bits/per channel state) are listed in Table \Rmnum{1}. \vspace*{-5mm}
\begin{table}[H]
\centering
\caption{Feedback rate for either receiver.}
\begin{tabular}{|c|c|c|c|c|}
\hline
 $\Delta$& $T$  & $\lceil\log_2(T+1)\rceil$  & Receiver 1  & Receiver 2  \\ \hline
 0.01&  461&  9 &  5.3 & 4.6  \\ \hline
 0.05&  60&  6 &  3.6 & 2.7 \\ \hline
% 0.1&  24&  5 &  2.7 & 1.9 \\ \hline
\end{tabular}
\end{table}
\vspace{-8mm} \noindent Compared with FLE which costs $\lceil\log_2(T+1)\rceil$ bits
per receiver per channel state, VLE can save almost half of the feedback bits.

%Additionally, Fig. \ref{fig:test1} shows that the slope of the curve for $q_{r}(\cdot)$ with
%$\Delta = 0.1$ in the high-$P$ region is smaller than that of the TDMA scheme, which
%indicates the achieved minimum rate of the TDMA scheme will eventually increase beyond
%that of $q_{r}(\cdot)$ with $\Delta = 0.1$ when $P$ is large enough. The reason
%originates from our diversity gain analysis in Lemma 4, which says the overall achieved
%diversity order is only $1/2$ for fixed $\Delta$, compared with $1$ achieved by the
%TDMA scheme. Among the three values of $\Delta$ we have simulated here, $0.1$ is
%large enough for $\Delta$ to not scale in proportion to $p^{-\frac{1}{3}}$. Although lack of
%adequate diversity gain, the much higher array gain ensures $q_{r}(\cdot)$ with a fixed
%$\Delta$ to still outperform the TDMA scheme a lot in the considered power region. We
%will elaborate this more in the simulations for outage probability.

In Fig. \ref{fig:test2}, we plot the rate losses of $q_r(\cdot)$ for different values of
$\Delta$ and the
 feedback rates $R_{r, {\rm VLE}, 1}$ and $R_{r, {\rm VLE}, 2}$. It shows
that the rate loss of $q_{r}(\cdot)$ decreases at least linearly with respect to $\Delta$ and
exponentially with $\min\left\{R_{r, {\rm VLE}, 1}, R_{r, {\rm VLE}, 2}\right\}$, which
validates the accuracy of our derived upper bounds in \eqref{rateBound_Linear} and
\eqref{rate_VLE}. In addition, Fig. \ref{fig:test2}(a) shows that $\Delta$ needs to be less
than $0.15$ such that $q_r(\cdot)$ can obtain a higher rate compared with the TDMA
scheme.

\begin{figure}
\centering
\begin{minipage}{.5\textwidth}
  \centering
  \includegraphics[width=.9\linewidth]{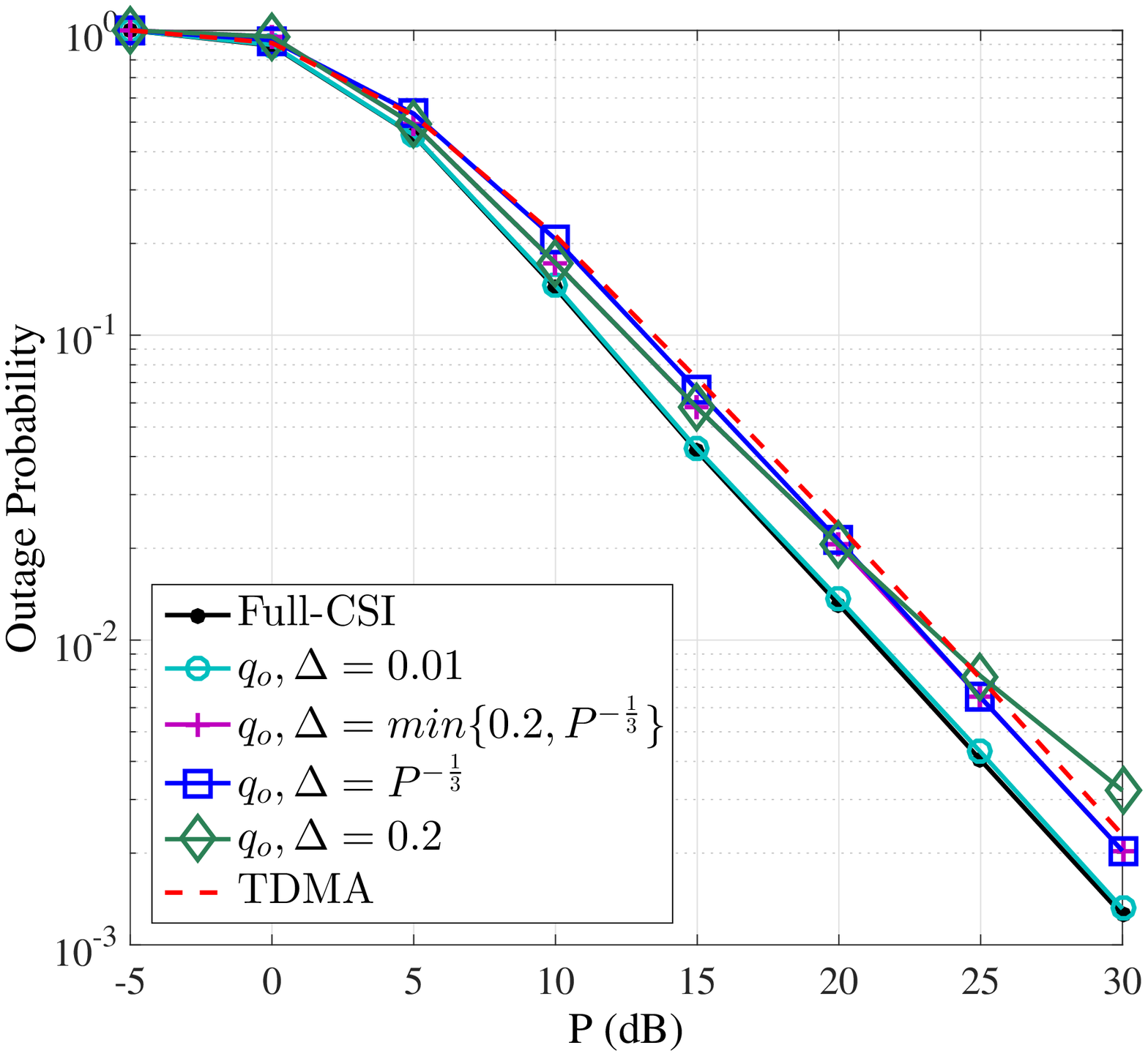}
  \captionof{figure}{Simulated outage probabilities of NOMA.}
  \label{fig:out1}
\end{minipage}%
\begin{minipage}{.5\textwidth}
  \centering
  \includegraphics[width=.9\linewidth]{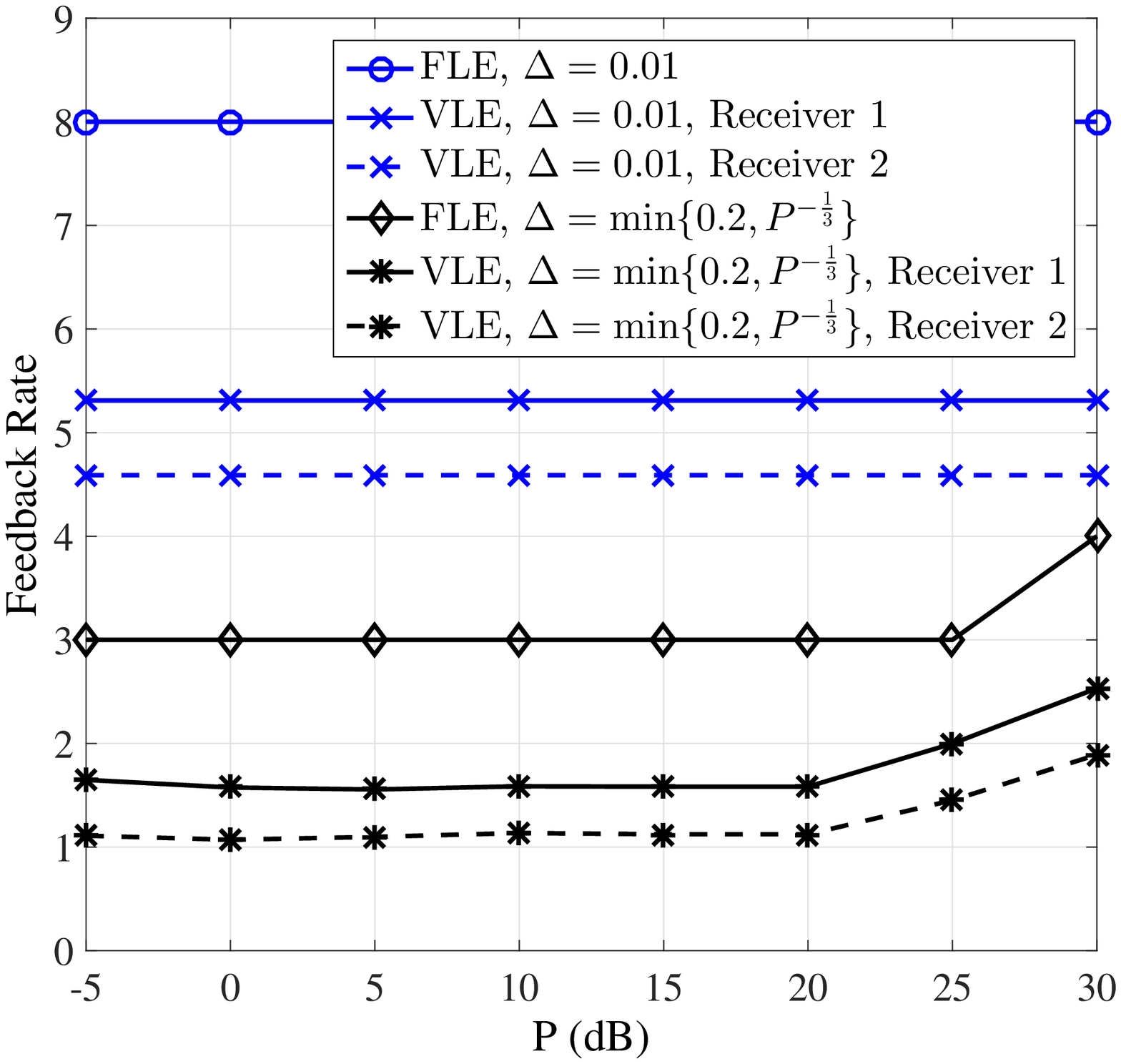}
  \captionof{figure}{Simulated feedback rates versus $P$.}
  \label{fig:out2}
\end{minipage}
\end{figure}

In Fig. \ref{fig:out1}, we compare the outage probabilities of the full-CSI case,
$q_{o}(\cdot)$ under various values of $\Delta$ and the TDMA scheme. It can be seen
that: (\rmnum{1}) The curve for $q_o(\cdot)$ with $\Delta = 0.01$ almost coincides with
that of the full-CSI case. (\rmnum{2}) When $P$ is large, $q_{o}(\cdot)$ with $\Delta =
0.2$ suffers from an insufficient diversity gain in the high-$P$ region. According to our
analysis in Lemma 4, $\Delta = 0.2$ is large enough not to scale with
$P^{-\frac{1}{3}}$.\footnote{The value $0.01$ for $\Delta$ will also exhibit an insufficient
diversity order as long as $P$ is large enough, although we might not be able to observe
this in the region of $P \leq 30$ dB in Fig. \ref{fig:out1}.} (\rmnum{3}) Although the
maximum diversity order is achieved when $\Delta = P^{-\frac{1}{3}}$, much less array
gain is obtained in the lower and medium-$P$ regions (where $\Delta$ is large).
Alternatively, $\Delta = \min\left\{0.2, P^{-\frac{1}{3}}\right\}$ will reserve both benefits
of the maximum diversity order brought by $P^{-\frac{1}{3}}$ and the higher array gain of
$\Delta = 0.2$.\footnote{We also observe a similar effect of $\Delta$ on the achieved
minimum rates, but we mainly elaborate it on outage probability.} The comparison of
feedback rates for VLE and FLE (which requires $\left\lceil\log_2(T+2)\right\rceil =
\left\lceil\log_2\left(\frac{\lambda_1}{2\Delta}\log\frac{1}{\Delta}+2\right)\right\rceil$
bits per channel state) under different values of $\Delta$ and $P$ is shown in Fig.
\ref{fig:out2}, which verifies the superiority of VLE. It can be seen that the feedback rates
for $\Delta = \min\left\{0.2, P^{-\frac{1}{3}}\right\}$ stay flat in the low and
medium-$P$ regions (since $0.2 \leq P^{-\frac{1}{3}}$). When $P^{-\frac{1}{3}} \leq 0.2$
where $P \geq 20.9$ dB, the feedback rates start to increase as $\Delta$ gets smaller.

\begin{figure}
\centering
\begin{minipage}{.5\textwidth}
  \centering
  \includegraphics[width=.9\linewidth]{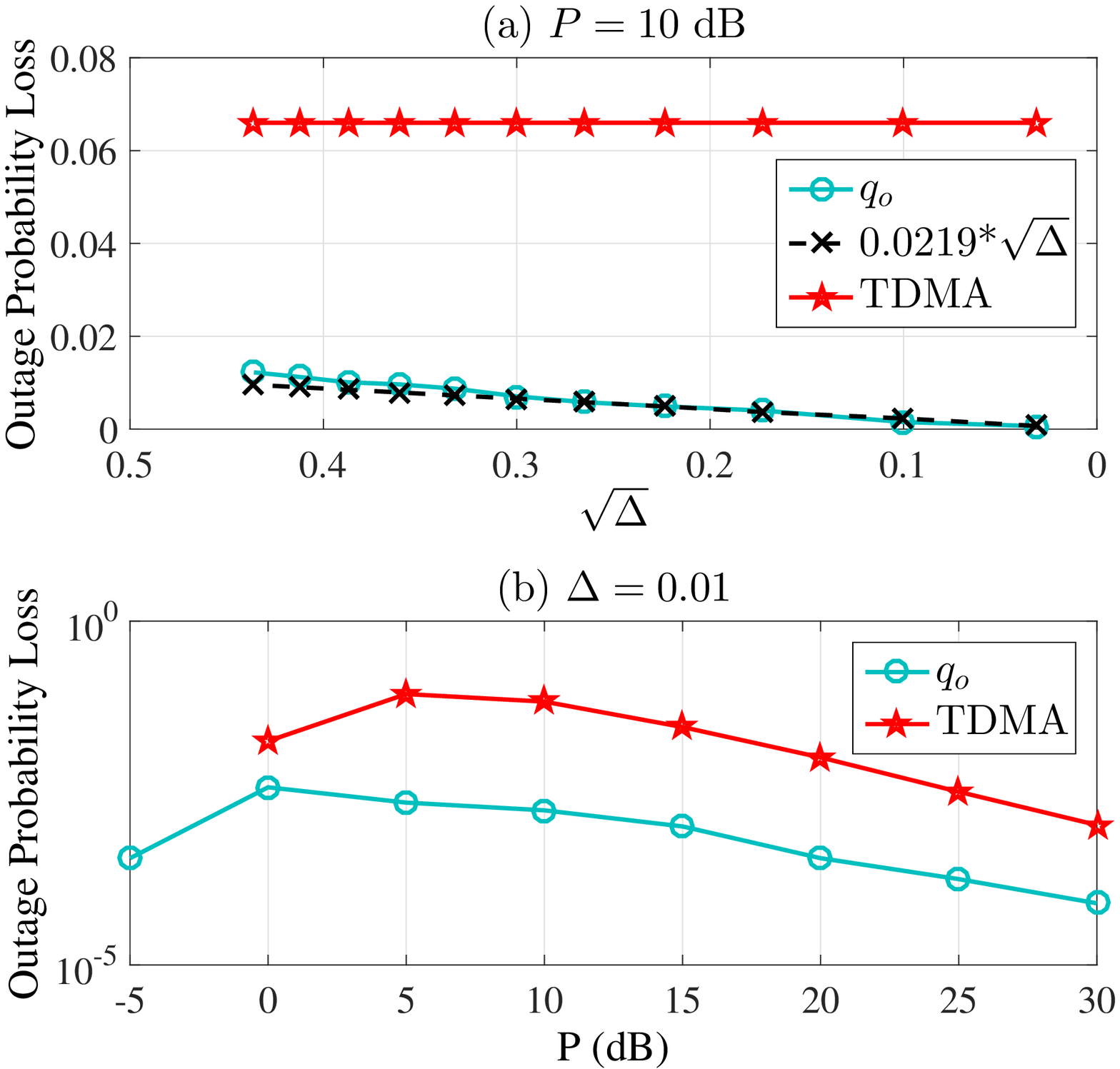}
  \captionof{figure}{Simulated outage probability losses versus $\sqrt{\Delta}$ and $P$.}
  \label{fig:out3}
\end{minipage}%
\begin{minipage}{.5\textwidth}
  \centering
  \includegraphics[width=.9\linewidth]{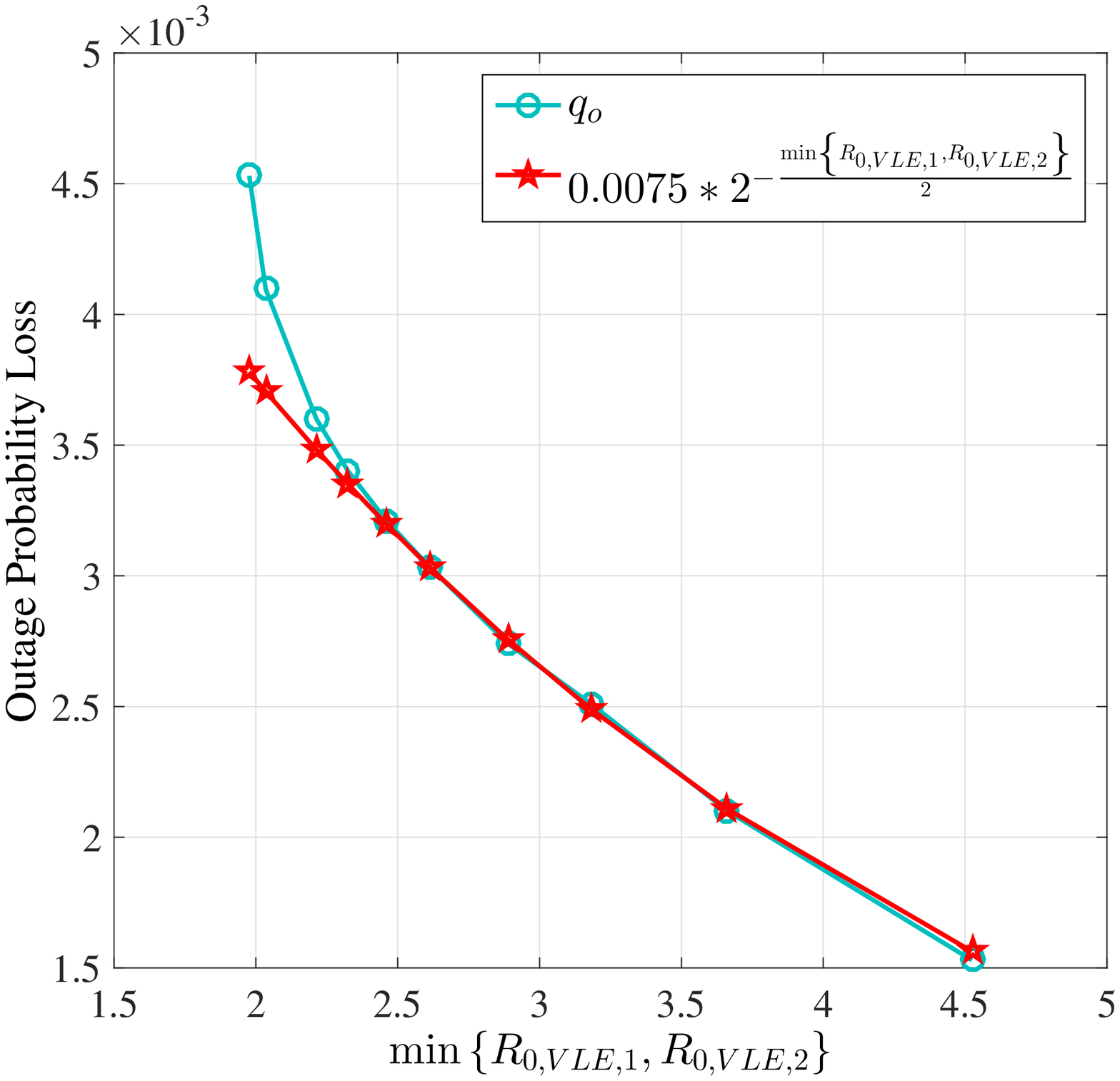}
  \captionof{figure}{Simulated outage probability losses versus
  $\min\left\{R_{0, {\rm VLE}, 1}, R_{0, {\rm VLE}, 2}\right\}$.}
  \label{fig:out4}
\end{minipage}
\end{figure}

In Fig. \ref{fig:out3}(a), the outage probability loss decays at least linearly with respect to
$\Delta$; in Fig. \ref{fig:out3}(b), the outage probability loss approaches zero whenever
$P\rightarrow 0$ or $P\rightarrow \infty$; in Fig. \ref{fig:out4}, the outage probability
loss decays at least exponentially with $\frac{\min\left\{R_{0, {\rm VLE}, 1}, R_{0, {\rm
VLE}, 2}\right\}}{2}$. All these observations validate our theoretical analysis.

\begin{figure}
\centering
\begin{minipage}{.5\textwidth}
  \centering
  \includegraphics[width=.9\linewidth]{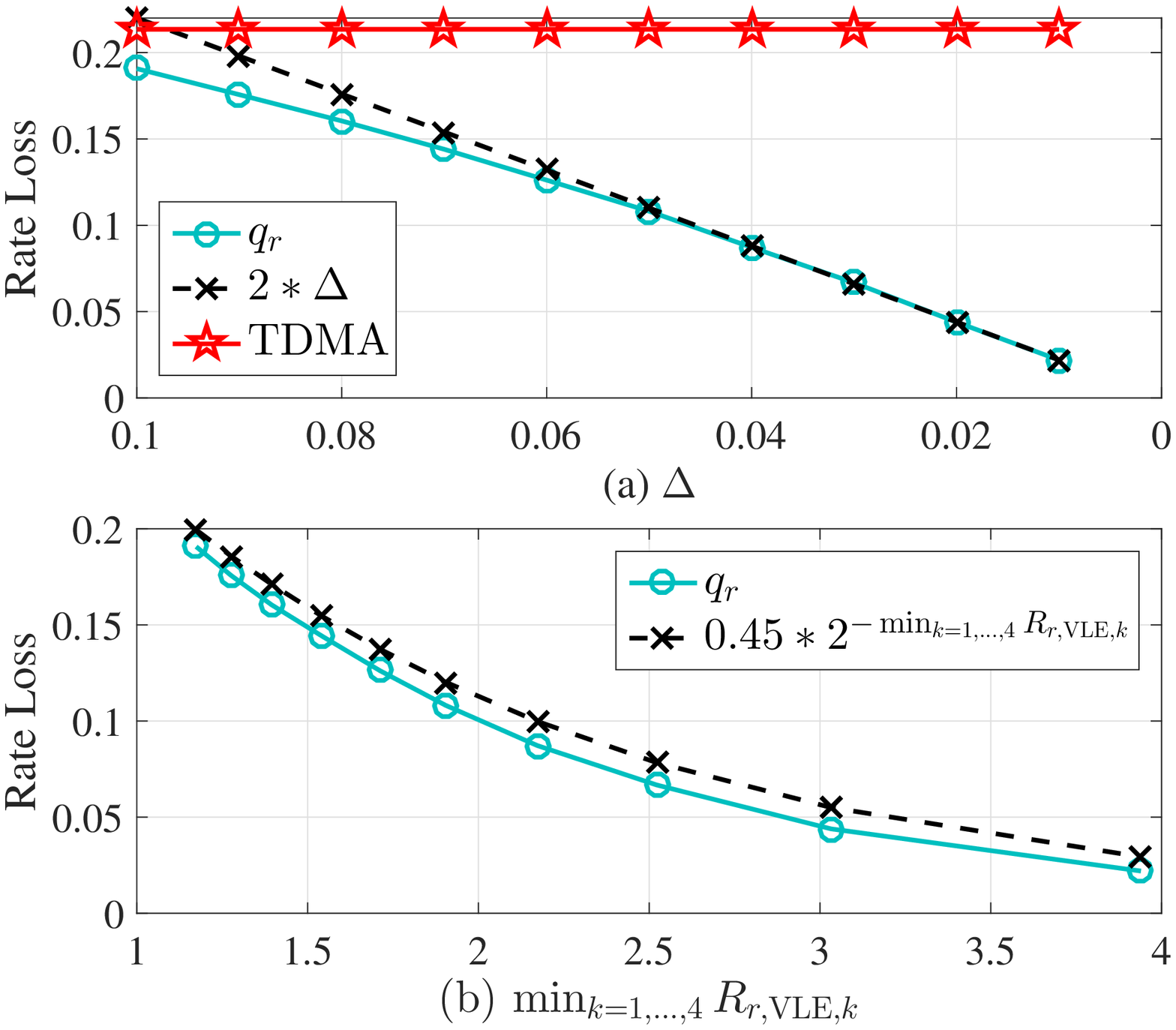}
  \captionof{figure}{Simulated rate losses versus (a) $\Delta$ and (b) $\min_{k = 1, \ldots, K} R_{r, {\rm VLE}, k}$ for $K = 4$ and $P = 10$ dB.}
  \label{fig:K4_1}
\end{minipage}%
\begin{minipage}{.5\textwidth}
  \centering
  \includegraphics[width=.9\linewidth]{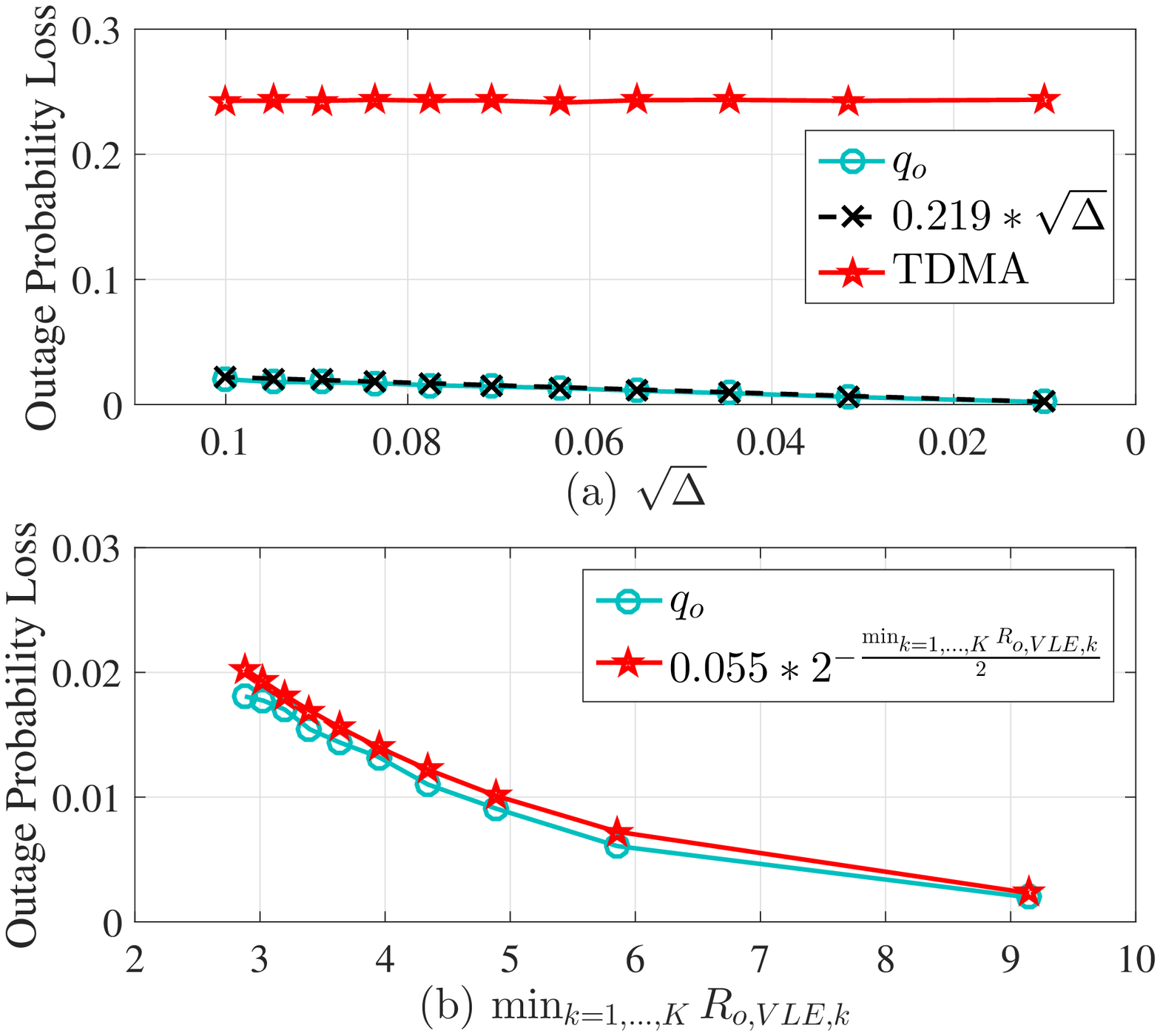}
  \captionof{figure}{Simulated outage probability losses versus (a) $\sqrt{\Delta}$ and (b) $\min_{k = 1, \ldots, K} R_{o, {\rm VLE}, k}$ for $K = 4$ and $P = 10$ dB.}
  \label{fig:K4_2}
\end{minipage}
\end{figure}

In Figs. \ref{fig:K4_1} and \ref{fig:K4_2}, we simulated the rate and outage
probability losses for more than two receivers. For Receiver $k$, the channel variance is set
to be $\lambda_k = \frac{1}{k}$, the maximum number of bins $T$ for $q_r(\cdot)$ and
$q_o(\cdot)$ is $T = \frac{\lambda_k}{\Delta}\log\frac{1}{\Delta}$, and the accuracy
used by the bisection method is $\epsilon = 10^{-4}$. We simply treat the result of bisection method
based on perfect CSI as the ``full-CSI'' performance. Compared with Figs. \ref{fig:test2},
\ref{fig:out3} and \ref{fig:out4} for $K = 2$, Figs. \ref{fig:K4_1} and
\ref{fig:K4_2} exhibit very similar relationships between the losses and $\Delta$ or the
feedback rates.

\section{Conclusions and Future Work}

We have introduced efficient quantizers for rate adaptation and outage probability of
minimum rate in NOMA with two receivers. We have proved that the losses in rate and
outage probability both decrease at least exponentially with the minimum of the feedback rates.
Furthermore, we generalized the proposed quantizers to NOMA with any number of
receivers. The limited feedback design for the MIMO-NOMA networks will be an
interesting future research direction.

\section*{Appendix A: Proof of Lemma 2}

To clarify, the notation $D_i$ for $i \in \mathds{N}$ represents a positive constant
independent of $P, T$ and $\Delta$. The average rate loss of $q_r(\cdot)$ can be
expressed as
\begin{align}
  \mathtt{E}\left[r_{{\rm loss}}\right] = \underbrace{\int_{\mathpzc{H}_{0, \geq}} r_{{\rm loss}} \prod_{i=1}^2 f_{H_i}(H_i)
  {\rm d}
  H_i}_{= \mathtt{E}_{\geq}\left[r_{{\rm loss}}\right]}  +
  \underbrace{\int_{\mathpzc{H}_{0, <}} r_{{\rm loss}} \prod_{i=1}^2 f_{H_i}(H_i)
  {\rm d}
  H_i}_{
  = \mathtt{E}_{<}\left[r_{{\rm loss}}\right]},\nonumber
\end{align}
where $\mathpzc{H}_{0, \geq} = \left\{(H_1, H_2): q_r(H_1) \geq q_r(H_2)\right\}$
and $\mathpzc{H}_{0, <} = \left\{(H_1, H_2): q_r(H_1) < q_r(H_2)\right\}$. We will
only show $\mathtt{E}_{\geq}\left[r_{{\rm loss}}\right] \leq \log_2\left(1 + D_0 \times P
\times \max\left\{e^{-\frac{T\Delta}{\lambda_1}}, \Delta\right\}\right)$, and skip the
proof for $\mathtt{E}_{<}\left[r_{{\rm loss}}\right]$ due to similarity. Note that $q_r(H_1) \geq q_r(H_2)$ does not necessarily mean $H_1 \geq H_2$, since it
is possible that $q_r(H_1) = q_r(H_2)$ and $H_1 < H_2$.  When $q_r(H_1) \geq
q_r(H_2)$, define
\begin{eqnarray}
\label{snr_max}
\begin{array}{l}
{\sf snr}_{\max} =
\left\{
\begin{matrix}
\alpha^{\star} H_1 = g_{\geq}(H_1, H_2), & \text{if } H_1 \geq H_2, \\
\alpha^{\star} H_2 = g_{<}(H_1, H_2), & \text{if } H_1 < H_2,
\end{matrix}
\right. \\
{\sf snr}_{q_r} = \alpha_{q_r} \times q_r(H_1) = g_{\geq}\left(q_r(H_1), q_r(H_2)\right),\\
{\sf snr}_{\rm loss} = {\sf snr}_{\max} - {\sf snr}_{q_r}.
\end{array}
\end{eqnarray}
where $g_{\geq}(x, y) = \frac{2xy}{\sqrt{(x+y)^2 + 4xy^2P}+x+y}$ and $g_{<}(x, y) =
\frac{2xy}{\sqrt{(x+y)^2+4x^2yP} + x + y}$. Then, we have $ r_{\rm loss} = \log_2\left(1 +
P \times {\sf snr}_{\max}\right) - \log_2\left(1 + P\times {\sf snr}_{q_r}\right)
 = \log_2\left(1 + P \frac{{\sf snr}_{\rm loss}}{1 + P\times
{\sf snr}_{q_r}}\right) \leq \log_2\left(1 + P\times {\sf snr}_{\rm loss}\right)$. Grounded on
this, the main steps of the proof are listed as follows:
\begin{enumerate}
  \item[(1)] Partition $\mathpzc{H}_{0, \geq}$ into the following mutually disjoint
      sub-regions $\mathpzc{H}_1, \ldots, \mathpzc{H}_4$:
\begin{eqnarray}
  \begin{array}{l}
    \mathpzc{H}_1 = \left\{(H_1, H_2): q_r(H_1) \geq q_r(H_2), H_1 < T\Delta, H_2 < T\Delta,
    H_1 < \Delta \text{ or } H_2 < \Delta\right\}, \nonumber\\
    \mathpzc{H}_2 = \left\{(H_1, H_2): q_r(H_1) \geq q_r(H_2), H_1 \geq H_2, \Delta \leq H_1 < T\Delta,
    \Delta \leq H_2 < T\Delta\right\} \nonumber\\
    \mathpzc{H}_3 = \left\{(H_1, H_2): q_r(H_1) = q_r(H_2), H_1 < H_2, \Delta \leq H_1 < T\Delta,
    \Delta \leq H_2 < T\Delta\right\} \nonumber\\
    \mathpzc{H}_4 = \left\{(H_1, H_2): q_r(H_1) \geq q_r(H_2), H_1 \geq T\Delta \text{ or } H_2 \geq T\Delta\right\}.
  \end{array}
\end{eqnarray}
Here, $\mathpzc{H}_1$ and $\mathpzc{H}_4$ are edge regions where $H_i < \Delta$
 or $H_i \geq T\Delta$; $\mathpzc{H}_2$ and $\mathpzc{H}_3$ are the dominant
 regions where $\Delta \leq H_i < T\Delta$. It can be verified that $\mathpzc{H}_i \cap
 \mathpzc{H}_j = \emptyset$ for $i \neq j$,
      and $\mathpzc{H}_{0, \geq} = \bigcup_{i=1}^4 \mathpzc{H}_i$.
\item[(2)] Let $\mathscr{E}_i = \int_{\mathpzc{H}_i}  {\sf snr}_{{\rm
    loss}}\prod_{i=1}^2 f_{H_i}(H_i)
  {\rm d} H_i$. Then,
    $\mathtt{E}_{\geq}\left[{\sf snr}_{\rm loss}\right] = \sum_{i = 1}^4 \mathscr{E}_i$.
    Prove $\mathscr{E}_i \leq D_i \times \max\left\{e^{-\frac{T\Delta}{\lambda_1}},
    \Delta\right\}$ for $i = 1,
     \ldots, 4$.
\item[(3)] After Steps (1) and (2), we obtain $\mathtt{E}_{\geq}\left[{\sf snr}_{\rm
    loss}\right] \leq D_0 \times \max\left\{e^{-\frac{T\Delta}{\lambda_1}},
    \Delta\right\}$. Based on Jensen's inequality, we have
\begin{align}
  \mathtt{E}_{\geq}\left[r_{\rm loss}\right]
  & \leq \mathtt{E}_{\geq} \left[\log_2\left(1 +
P\times{\sf snr}_{\rm loss}\right)\right] \leq \log_2\left(1 + P \times \mathtt{E}_{\geq}\left[{\sf snr}_{\rm
    loss}\right] \right)  \leq \log_2\left(1 + D_0 \times P \times
\max\left\{e^{-\frac{T\Delta}{\lambda_1}}, \Delta\right\}\right).\nonumber
\end{align}
\end{enumerate}
Now, we only need to show the upper bound on $\mathscr{E}_i$ in Step (2).

For $\mathscr{E}_1$, since $\mathpzc{H}_1 \subseteq \left\{(H_1, H_2): H_2 \leq \Delta
\right\}$ and ${\sf snr}_{{\rm loss}} \leq {\sf snr}_{{\max}} \leq H_1$, we obtain
\begin{align}
%\label{upperBound_E1}
\mathscr{E}_1 \leq \int_0^{\infty} H_1\frac{e^{-\frac{H_1}{\lambda_1}}}{\lambda_1} {\rm d}H_1
\int_0^{\Delta} \frac{e^{-\frac{H_2}{\lambda_2}}}{\lambda_2} {\rm d}H_2
  = \lambda_1  \left(1 - e^{-\frac{\Delta}{\lambda_2}}\right) \leq \lambda_1
  \times \frac{\Delta}{\lambda_2} = {D_1 \times \Delta}, \nonumber
  \end{align}
where the last inequality follows since $1-e^{-x} \leq x$ for $x \geq 0$.

For $\mathscr{E}_2$, since $H_1 \geq H_2$ and
  $q_r\left({H_i}\right) \leq H_i \leq q_r\left({H_i}\right) + \Delta$ for $H_i \leq T\Delta$, we upper-bound ${\sf snr}_{\rm loss}$ by
  \begin{align}
  \label{snr_loss_upperbound}
   {\sf snr}_{\rm loss}
    & = \frac{2H_1 H_2}{\underbrace{\sqrt{\left(H_1 + H_2\right)^2 + 4 H_1 H_2^2 P}}_{=\Upsilon} + \left(H_1 + H_2\right)} \nonumber\\
    & - \frac{2 q_r\left(H_1\right)q_r\left(H_2\right)}{\underbrace{\sqrt{\left[q_r\left(H_1\right) + q_r\left(H_2\right)\right]^2 + 4 q_r\left(H_1\right)
    q_r^2\left(H_2\right) P} + \left[q_r\left(H_1\right) + q_r\left(H_2\right)\right]}_{ \leq {\Upsilon} + H_1 + H_2}} \nonumber\\
    & \leq 2\frac{H_1 H_2 - q_r\left(H_1\right)q_r\left(H_2\right)}{\Upsilon + H_1 + H_2}
    \leq 2\frac{H_1 H_2 - \left(H_1 - \Delta\right)\left(H_2 - \Delta\right)}{\Upsilon + H_1 + H_2}
   = 2\Delta\frac{H_1 + H_2 - \Delta}{\Upsilon + H_1 + H_2}
    \leq 2 \Delta.
  \end{align}
  Then, an upper bound on $\mathscr{E}_2$ can be $\mathscr{E}_2
     \leq 2 \Delta \int_{\mathcal{H}_2} \prod_{i=1}^2 f_{H_i}(H_i)
  {\rm d}
  H_i
     \leq 2 \Delta = D_2 \times \Delta$.

  For $\mathscr{E}_3$, we have $q_r(H_1) = q_r(H_2) \leq H_1 < H_2$ and $q_r\left({H_i}\right)
  \leq H_i \leq q_r\left({H_i}\right) + \Delta$ hold for $(H_1, H_2) \in \mathpzc{H}_3$.
  Similar to \eqref{snr_loss_upperbound}, we can also obtain ${\sf snr}_{\rm loss} \leq 2\Delta$ and $\mathscr{E}_3 \leq D_3 \times \Delta$.

   For $\mathscr{E}_4$, since $\mathpzc{H}_4 \subseteq \left\{\left(H_1, H_2\right): H_1 > T\Delta\right\}$ and
    ${\sf snr}_{\rm loss} \leq {\sf snr}_{\max} \leq H_2$,
   the upper-bound on $\mathscr{E}_4$ can be
   \begin{align}
   % \label{upperBound_E4}
     \mathscr{E}_4
     & \leq  \int_{T\Delta}^{\infty} f_{H_1}(H_1) {\rm d}H_1 \int_{0}^{\infty}H_2 f_{H_2}(H_2)  {\rm d}H_2
      = \int_{T\Delta}^{\infty} \frac{e^{-\frac{H_1}{\lambda_1}}}{\lambda_1} {\rm d}H_1
     \int_0^{\infty} H_2 \frac{e^{-\frac{H_2}{\lambda_2}}}{\lambda_2}  {\rm d}H_2 = \lambda_2 e^{-\frac{T\Delta}{\lambda_1}}
      = D_4 \times e^{-\frac{T\Delta}{\lambda_1}}. \nonumber
   \end{align}
   We have accomplished Step (2) and the proof of \eqref{rateBound} is complete.
\hfill$\blacksquare$

\section*{Appendix B: Proof of Lemma 3}

When the uniform quantizer $q_o(\cdot)$ is applied, the outage probability loss in
\eqref{loss} is rewritten as
\begin{align}
  \mathtt{out}_{{\rm loss}, q_{o}}
  & = \underbrace{\int_{\mathpzc{I}_{0, \geq}}
   {\pmb 1}_{\min\left\{r_1\left(\alpha_{q_o}\right), r_2\left(\alpha_{q_o}\right)\right\}
    < r_{\rm th}}
   \prod_{i=1}^2 f_{H_i}(H_i)
  {\rm d}
  H_i}_{= \mathtt{out}_{\geq, {\rm loss}, q_{o}}} + \underbrace{\int_{\mathpzc{I}_{0, <}}
   {\pmb 1}_{\min\left\{r_1\left(\alpha_{q_o}\right), r_2\left(\alpha_{q_o}\right)\right\} < r_{\rm th}}
   \prod_{i=1}^2 f_{H_i}(H_i)
  {\rm d}
  H_i}_{= \mathtt{out}_{<, {\rm loss}, q_{o}}}.\nonumber
\end{align}
where
\begin{eqnarray}
  \begin{array}{l}
    \mathpzc{I}_{0, \geq} = \left\{(H_1, H_2):    q_r(H_1) \geq q_r(H_2),
r_{\max}
= \log_2\left(1 + P\times {\sf snr}_{\max} \right)\geq r_{\rm
th}\right\} \nonumber\\
\quad \quad =
\left\{(H_1, H_2):        q_r(H_1) \geq q_r(H_2),
{\sf snr}_{\max} \geq \frac{\beta}{P} = \frac{2^{r_{\rm th}}-1}{P}\right\}, \nonumber\\
%\quad  \quad =
%\left\{(H_1, H_2):
%g_{\geq}(H_1, H_2) \geq \frac{\beta}{P}, q_r(H_1) \geq q_r(H_2), H_1 \geq H_2\right\} \nonumber\\
%\quad\quad  \quad
%\bigcup \left\{(H_1, H_2):
%g_{<}(H_1, H_2) \geq \frac{\beta}{P}, q_r(H_1) = q_r(H_2), H_1 < H_2\right\}, \\
    \mathpzc{I}_{0, <} = \left\{(H_1, H_2):   q_r(H_1) < q_r(H_2),
{\sf snr}_{\max} < \frac{\beta}{P}\right\}. \nonumber
  \end{array}
\end{eqnarray}
and ${\sf snr}_{\max}$ is defined in \eqref{snr_max}. We show $\mathtt{out}_{\geq,
{\rm loss}, q_{o}} \leq D_5 \times e^{-\frac{D_{6}}{P}} \times \frac{1 + \sqrt{P}}{P} \times
\max\left\{e^{-\frac{T\Delta}{\lambda_1}}, {\Delta}^{\frac{1}{2}}, \Delta^{\frac{3}{2}}\right\}$ and skip the
proof for $\mathtt{out}_{<, {\rm loss}, q_{o}}$ due to similarity. The
main steps of the proof are:
\begin{enumerate}
  \item[(1)] Partition $\mathpzc{I}_{0, \geq}$ into the following mutually disjoint
      sub-regions:
      \begin{eqnarray}
        \begin{array}{l}
           \mathpzc{I}_1  = \left\{(H_1, H_2): q_r(H_1) \geq q_r(H_2), {\sf snr}_{\max} \geq \frac{\beta}{P},
  H_1 \leq \Delta, H_2 \leq \Delta\right\},\nonumber\\
  \mathpzc{I}_2  = \left\{(H_1, H_2): q_r(H_1) \geq q_r(H_2), {\sf  snr}_{\max} = g_{\geq}(H_1, H_2) \geq \frac{\beta}{P},
  \Delta < H_1 \leq T\Delta, H_2 \leq \Delta\right\},\nonumber\\
  \mathpzc{I}_3 = \left\{(H_1, H_2): q_r(H_1) \geq q_r(H_2), H_1 \geq H_2, g_{\geq}(H_1, H_2) \geq \frac{\beta}{P},
  \Delta < H_1 \leq T\Delta, \Delta < H_2 \leq T\Delta \right\},\nonumber\\
  \mathpzc{I}_4 = \left\{(H_1, H_2): q_r(H_1) = q_r(H_2), H_1 < H_2, g_{<}(H_1, H_2) \geq \frac{\beta}{P},
  \Delta < H_1 \leq T\Delta, \Delta < H_2 \leq T\Delta \right\},\nonumber\\
  \mathpzc{I}_5 = \left\{(H_1, H_2): q_r(H_1) \geq q_r(H_2), {\sf snr}_{\max} \geq \frac{\beta}{P},
  H_1 > T\Delta  \text{ or } H_2 > T\Delta \right\}.\nonumber
  \end{array}
      \end{eqnarray}
      Here, $\mathpzc{I}_1$, $\mathpzc{I}_2$ and $\mathpzc{I}_5$ are the marginal
      regions where $H_i \leq \Delta$ or $H_i > T\Delta$; $\mathpzc{I}_3$ and
      $\mathpzc{I}_4$ are the main regions where $\Delta < H_i \leq T\Delta$. It can be verified that $\mathpzc{I}_i \cap
 \mathpzc{I}_j = \emptyset$ for $i \neq j$,
      and $\mathpzc{I}_{0, \geq} = \bigcup_{i=1}^5 \mathpzc{I}_i$.
      \item[(2)] Let $\mathscr{F}_i = \int_{\mathpzc{I}_i} {\pmb
          1}_{\min\left\{r_1\left(\alpha_{q_o}\right), r_2\left(\alpha_{q_o}\right)\right\} <
          r_{\rm th}} \prod_{i=1}^2 f_{H_i}(H_i)
  {\rm d} H_i$. Then,
          $\mathtt{out}_{\geq, {\rm loss}, q_{o}} = \sum_{i=1}^5 \mathscr{F}_i$. Prove
          $\mathscr{F}_i \leq D_{2i+5} \times e^{-\frac{D_{{2i+6}}}{P}} \times \frac{1 +
          \sqrt{P}}{P} \times \max\left\{e^{-\frac{T\Delta}{\lambda_1}}, {\Delta}^{\frac{1}{2}}, \Delta^{\frac{3}{2}}\right\}$ for $i = 1, \ldots, 5$.
\end{enumerate}

Now, we need to show the upper bound on $\mathscr{F}_i$ in Step (2).

For $\mathscr{F}_1$, we have $q_o(H_1) = q_o(H_2) = \Delta \geq H_2$, and thus,
$\alpha_{q_o} = \frac{1}{\sqrt{P\Delta+1}+1} \leq \frac{1}{\sqrt{PH_2+1}+1}$. For any
$(H_1, H_2) \in \mathpzc{I}_1$, since $g_{\geq}(x, y) \leq \min\{x, y\}$ and $g_{<}(x, y)
\leq \min\{x, y\}$, it must have $\frac{\beta}{P} \leq {\sf snr}_{\max} \leq \min\{H_1,
H_2\}$. Moreover, we obtain ${\pmb 1}_{\min\left\{r_1\left(\alpha_{q_o}\right),
r_2\left(\alpha_{q_o}\right)\right\}
  < r_{\rm th}}
  \leq {\pmb 1}_{r_1\left(\alpha_{q_o}\right) < r_{\rm th}} + {\pmb
1}_{r_2\left(\alpha_{q_o}\right) < r_{\rm th}}$, and
\begin{eqnarray}
  \begin{array}{l}
    {\pmb
1}_{r_1\left(\alpha_{q_o}\right) < r_{\rm th}} = {\pmb 1}_{H_1 \times \alpha_{q_o} <
\frac{\beta}{P}} = {\pmb 1}_{H_1 < \beta\frac{\sqrt{P\Delta+1}+1}{P}}, \nonumber\\
{\pmb
1}_{r_2\left(\alpha_{q_o}\right) < r_{\rm th}}
=
{\pmb
1}_{\frac{H_2\left(1-\alpha_{q_o}\right)}{PH_2\alpha_{q_o}+1} < \frac{\beta}{P}}
\leq
{\pmb
1}_{\frac{H_2\left(1-\frac{1}{\sqrt{PH_2+1}+1}\right)}{PH_2\times\frac{1}{\sqrt{PH_2+1}+1}+1} < \frac{\beta}{P}}
= {\pmb 1}_{H_2 <\frac{\beta^2+2\beta}{P}}. \nonumber
  \end{array}
\end{eqnarray}
Thus, an upper bound on $\mathscr{F}_1$ is
   \begin{align}
   \mathscr{F}_1
   & \leq
   \int_{\mathpzc{I}_1} {\pmb 1}_{H_1 < \beta\frac{\sqrt{P\Delta+1}+1}{P}} \prod_{i=1}^2 f_{H_i}(H_i)
  {\rm d}
  H_i
   +
   \int_{\mathpzc{I}_1} {\pmb 1}_{H_2 < \frac{\beta^2 + 2\beta}{P}} \prod_{i=1}^2 f_{H_i}(H_i)
  {\rm d}
  H_i
   \nonumber\\
   & \leq
   \int_{\frac{\beta}{P}}^{\beta\frac{\sqrt{P\Delta+1}+1}{P}} \frac{e^{-\frac{H_1}{\lambda_1}}}{\lambda_1}
   \int_{\frac{\beta}{P}}^{\Delta} \frac{e^{-\frac{H_2}{\lambda_2}}}{\lambda_2}
  {\rm d}H_1 {\rm d}H_2
   +
   \int_{\frac{\beta}{P}}^{\Delta} \frac{e^{-\frac{H_1}{\lambda_1}}}{\lambda_1}
   \int_{\frac{\beta}{P}}^{\frac{\beta^2 + 2\beta}{P}} \frac{e^{-\frac{H_2}{\lambda_2}}}{\lambda_2}
  {\rm d}H_1 {\rm d}H_2
  \nonumber\\
  & \leq
  \frac{e^{-\frac{\frac{\beta}{P}}{\lambda_1}}}{\lambda_1} \times\left[\beta\frac{\sqrt{P\Delta+1}+1}{P}-\frac{\beta}{P}\right]
  \times \frac{1}{\lambda_2} \times \left[\Delta - \frac{\beta}{P}\right]
  +
  \frac{1}{\lambda_1} \times \left[\Delta - \frac{\beta}{P}\right]
  \times
  \frac{e^{-\frac{\frac{\beta}{P}}{\lambda_2}}}{\lambda_2} \times\left[\frac{\beta^2 + 2\beta}{P}-\frac{\beta}{P}\right]
  \nonumber\\
  & \leq
  \frac{e^{-\frac{\frac{\beta}{P}}{\lambda_1}}}{\lambda_1} \times \beta\times\overbrace{\frac{\sqrt{P\Delta+1}}{P}}^
  {\leq \sqrt{P\Delta}+1}
  \times \frac{1}{\lambda_2} \times \Delta
  +
  \frac{1}{\lambda_1} \times \Delta
  \times
  \frac{e^{-\frac{\frac{\beta}{P}}{\lambda_2}}}{\lambda_2} \times \frac{\beta^2 + \beta}{P}
  \nonumber\\
  \label{F1_upperBound}
  & \leq D_{17} \times
e^{-\frac{D_{{18}}}{P}}
\times \frac{\sqrt{P\Delta}+1}{P} \times \Delta
+ D_{19} \times
e^{-\frac{D_{{20}}}{P}} \times\frac{\Delta}{P}.
   \end{align}

    For $\mathscr{F}_2$, let $\mathscr{F}_{2, i} = \int_{\mathpzc{I}_2} {\pmb 1}_{r_i\left(\alpha_{q_o}\right)
    < r_{\rm th}} \prod_{i=1}^2 f_{H_i}(H_i)
  {\rm d}
  H_i$ for $i = 1, 2$. Then, $\mathscr{F}_2 \leq
    \mathscr{F}_{2, 1} + \mathscr{F}_{2, 2}$. For $\mathscr{F}_{2, 1}$, since $H_1 > H_2$ for $(H_1, H_2) \in \mathpzc{I}_2$ and $g_{\geq}(x, y)$ is
 increasing on $x$ and $y$, we have
     \begin{align}
     {\pmb 1}_{r_1\left(\alpha_{q_o}\right) < r_{\rm th}}
     \label{exp01}
     & =
     {\pmb 1}_{\frac{2 H_1 \times q_{o}\left(H_2\right) }{\sqrt{\left[q_{o}\left(H_1\right) + q_{o}\left(H_2\right)\right]^2 + 4 q_{o}\left(H_1\right) q_{o}^2\left(H_2\right) P} + \left[q_{o}\left(H_1\right) + q_{o}\left(H_2\right)\right]} < \frac{\beta}{P}} \nonumber\\
     & \leq
     {\pmb 1}_{\frac{2 \left(q_{o}(H_1) - \Delta\right) \times q_{o}\left(H_2\right) }{\sqrt{\left[q_{o}\left(H_1\right) + q_{o}\left(H_2\right)\right]^2 + 4 q_{o}\left(H_1\right) q_{o}^2\left(H_2\right) P}
   + \left[q_{o}\left(H_1\right) + q_{o}\left(H_2\right)\right]} < \frac{\beta}{P}}
   = {\pmb 1}_{g_{\geq}\left(q_o(H_1), q_o(H_2)\right) < \frac{\beta}{P}\times\frac{1}{1 - \frac{\Delta}{q_o(H_1)}}} \\
   \label{exp02}
   &
   \leq {\pmb 1}_{g_{\geq}\left(q_o(H_1), q_o(H_2)\right) < \frac{\beta}{P}\times\left({1 + \frac{2\Delta}{q_o(H_1)}}\right)}
   \leq   {\pmb 1}_{g_{\geq}\left(q_o(H_1), q_o(H_2)\right) < \frac{\beta}{P}\times\left({1 + \frac{2\Delta}{q_o(H_2)}}\right)} \\
   \label{exp03}
   & \leq {\pmb 1}_{g_{\geq}\left(q_o(H_1), q_o(H_2)\right) < \frac{\beta}{P}\times\left({1 + \frac{2\Delta}{H_2}}\right)}
   \leq {\pmb 1}_{g_{\geq}\left(H_1, H_2\right) < \frac{\beta}{P}\times\left({1 + \frac{2\Delta}{H_2}}\right)},
     \end{align}
     where \eqref{exp01} follows from $q_o(H_1) \leq H_1 + \Delta$, \eqref{exp02} follows from
      $\left(1 - \frac{\Delta}{q_o(H_1)}\right) \times \left(1 + \frac{2\Delta}{q_o(H_1)}\right) \geq 1$ because
     $q_o(H_1) \geq 2\Delta > q_o(H_2) = \Delta$, and \eqref{exp03} follows from $q_o(H_2) \geq H_2$
     and $g_{\geq}\left(q_o(H_1), q_o(H_2)\right) \geq g_{\geq}\left(H_1, H_2\right)$. Then, we obtain $\mathscr{F}_{2, 1} \leq
     \int_{\mathpzc{I}_2^{'} = \mathpzc{I}_2 \cap \left\{(H_1, H_2): g_{\geq}\left(H_1, H_2\right)<
     \frac{\beta}{P}\times\left({1 + \frac{2\Delta}{H_2}}\right)\right\}} \prod_{i=1}^2 f_{H_i}(H_i)
  {\rm d}
  H_i$.

     We change the integration variables from $(H_1, H_2)$ to $(\phi, H_2)$ where $\phi = g_{\geq}(H_1, H_2)$.
     Then, $H_1 = \frac{\phi^2P + \phi}{H_2 - \phi}\times H_2$, and the Jacobian matrix is
     $\left|\frac{{\rm d}H_1}{{\rm d}\phi}\right| =\frac{2\phi P H_2 + H_2 - \phi^2 P}{\left(H_2 - \phi\right)^2} \times H_2
     \leq \frac{2\phi P H_2 + H_2}{\left(H_2 - \phi\right)^2} \times H_2 \leq \frac{2\phi P H_2 + 2 H_2}{\left(H_2 - \phi\right)^2} \times H_2
     = \frac{2\left(\phi P +1 \right)}{\left(H_2 - \phi\right)^2} \times H_2^2$.
     For any $(H_1, H_2) \in \mathpzc{I}_2^{'}$, we have: (\rmnum{1}) $\frac{\beta}{P} \leq \phi = g_{\geq}(H_1, H_2) \leq H_2$ and
      $\phi < \frac{\beta}{P}\times\left({1 + \frac{2\Delta}{H_2}}\right)$; (\rmnum{2}) since $H_1 \geq H_2$,
      $H_1 = \frac{\phi^2P + \phi}{H_2 - \phi}\times H_2 \geq H_2$, then, $H_2 \leq \phi^2P + 2\phi$. Therefore, $\mathscr{F}_{2, 1}$
      is derived as $\mathscr{F}_{2, 1} \leq
     \int_{\mathpzc{I}_2^{''} =
     \left\{
       (H_1, H_2): \frac{\beta}{P} \leq H_2 \leq \phi^2P + 2\phi,
       \frac{\beta}{P} \leq \phi \leq \min\left\{H_2, \frac{\beta}{P}\left(1 + \frac{2\Delta}{H_2}\right)\right\}
     \right\}
     } \prod_{i=1}^2 f_{H_i}(H_i)
  {\rm d}
  H_i$.
  \begin{figure}
\centering
  \includegraphics[width=2in]{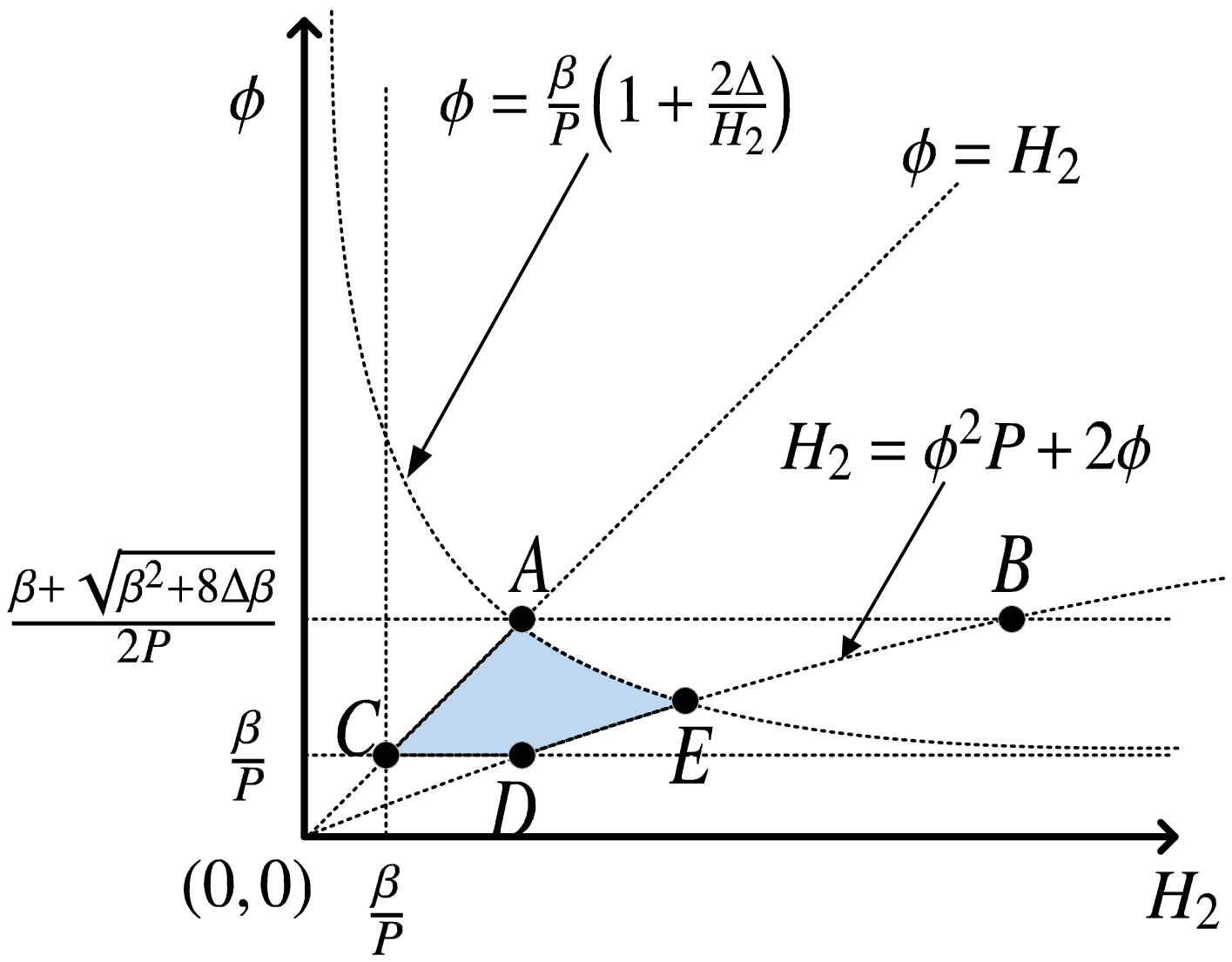}
  \caption{The integration region $\mathpzc{I}_2^{''}$.}
  \label{IntegrationArea}
\end{figure}
      The integration region $\mathpzc{I}_2^{''}$ is demonstrated in Fig. \ref{IntegrationArea} as the shaded area surrounded by the points $A, E, D$ and $C$.
      It can be strictly proven that $\mathpzc{I}_2^{''}$ is within the region surrounded the
     points $A, B, D$ and $C$. Recall that $H_1 = \frac{\phi^2P+\phi}{H_2 - \phi}\times H_2$ and $\left|\frac{{\rm d}H_1}{{\rm d}
     \phi}\right| \leq \frac{2\left(\phi P +1 \right)}{\left(H_2 - \phi\right)^2} \times H_2^2$. Then, we have
     \begin{align}
       \mathscr{F}_{2, 1}
       & \leq \int_{\frac{\beta}{P}}^{\frac{\beta+\sqrt{\beta^2 + 8\Delta\beta}}{2P}}
       \int_{\phi}^{\phi^2P + 2\phi}
       \frac{e^{-\frac{H_2}{\lambda_2}}}{\lambda_2} \times
       \frac{e^{-\frac{1}{\lambda_1} \times \frac{\phi^2P+\phi}{H_2 - \phi}\times H_2}}{\lambda_1}
       \times \frac{2\left(\phi P +1 \right)}{\left(H_2 - \phi\right)^2} \times H_2^2
       {\rm d}\phi {\rm d}H_2
        \nonumber\\
        & \stackrel{z = H_2 - \phi}{=} D_{21}
        \int_{\frac{\beta}{P}}^{\frac{\beta+\sqrt{\beta^2 + 8\Delta\beta}}{2P}}
       \int_{0}^{\phi^2P + \phi}
        \underbrace{e^{-\frac{z}{\lambda_2}-\frac{\phi}{\lambda_2}}}_{
        \leq e^{-\frac{z}{\lambda_2}} \times e^{-\frac{\frac{\beta}{P}}{\lambda_2}}} \times
       \underbrace{e^{-\frac{1}{\lambda_1} \times \frac{\phi^2P+\phi}{z}\times (z+\phi)}}_{
       \leq e^{-\frac{\phi^2\left(\phi P+ 1\right)}{\lambda_1 z}}}
       \times \frac{\phi P +1}{z^2} \times (z + \phi)^2
       {\rm d}\phi {\rm d}z
        \nonumber\\
        & \leq
        D_{21} \times e^{-\frac{\beta}{\lambda_2 P}}
        \int_{\frac{\beta}{P}}^{\frac{\beta+\sqrt{\beta^2 + 8\Delta\beta}}{2P}}
       \int_{0}^{\phi^2P + \phi}
       e^{-\frac{z}{\lambda_2}}
       e^{-\frac{\phi^2\left(\phi P+ 1\right)}{\lambda_1 z}}
       \times \left(\phi P + 1\right)
       \times \left[1 + \frac{2\phi}{z} + \frac{\phi^2}{z^2}\right]
       {\rm d}\phi {\rm d}z
        \nonumber\\
        & =
        D_{21} \times e^{-\frac{\beta}{\lambda_2 P}}
        \int_{\frac{\beta}{P}}^{\frac{\beta+\sqrt{\beta^2 + 8\Delta\beta}}{2P}}
       \int_{0}^{\phi^2P + \phi}
       e^{-\frac{z}{\lambda_2}}
       \underbrace{e^{-\frac{\phi^2\left(\phi P+ 1\right)}{\lambda_1 z}}}_{\leq 1}
       \times {\left(\phi P + 1\right)}
       {\rm d}\phi {\rm d}z
       \nonumber\\
       & +
       2 \times D_{21} \times e^{-\frac{\beta}{\lambda_2 P}}
        \int_{\frac{\beta}{P}}^{\frac{\beta+\sqrt{\beta^2 + 8\Delta\beta}}{2P}}
       \int_{0}^{\phi^2P + \phi}
       e^{-\frac{z}{\lambda_2}}
       e^{-\frac{\phi^2\left(\phi P+ 1\right)}{\lambda_1 z}}
       \times \left(\phi P + 1\right) \phi \times z^{-1}
       {\rm d}\phi {\rm d}z
       \nonumber\\
       & +
        D_{21} \times e^{-\frac{\beta}{\lambda_2 P}}
        \int_{\frac{\beta}{P}}^{\frac{\beta+\sqrt{\beta^2 + 8\Delta\beta}}{2P}}
       \int_{0}^{\phi^2P + \phi}
       e^{-\frac{z}{\lambda_2}}
       e^{-\frac{\phi^2\left(\phi P+ 1\right)}{\lambda_1 z}}
       \times \left(\phi P + 1\right) \phi^2 \times z^{-2}
       {\rm d}\phi {\rm d}z
       \nonumber\\
       & \leq
        D_{21} \times e^{-\frac{\beta}{\lambda_2 P}}
        \int_{\frac{\beta}{P}}^{\frac{\beta+\sqrt{\beta^2 + 8\Delta\beta}}{2P}} (\phi P + 1)
        {\rm d}\phi
       \int_{0}^{\infty}
       e^{-\frac{z}{\lambda_2}} {\rm d}z
       \nonumber\\
       & +
       \label{explain_a0}
       2 \times D_{21} \times e^{-\frac{\beta}{\lambda_2 P}}
        \int_{\frac{\beta}{P}}^{\frac{\beta+\sqrt{\beta^2 + 8\Delta\beta}}{2P}}
        \left(\phi P + 1\right) \phi
       \underbrace{\int_{0}^{\infty}
       e^{-\frac{z}{\lambda_2}}
       e^{-\frac{\phi^2\left(\phi P+ 1\right)}{\lambda_1 z}}
       \times  z^{-1}
       {\rm d}z}_{ = 2 {\mathcal{{K}}}_0\left(2\sqrt{\frac{\phi^2\left(\phi P+1\right)}{\lambda_1\lambda_2}}\right)
       \leq \frac{2\sqrt{\lambda_1\lambda_2}}{\phi\sqrt{\phi P+1}}}
       {\rm d}\phi
       \\
       \label{explain_a1}
       & +
        D_{21} \times e^{-\frac{\beta}{\lambda_2 P}}
        \int_{\frac{\beta}{P}}^{\frac{\beta+\sqrt{\beta^2 + 8\Delta\beta}}{2P}}
        \left(\phi P + 1\right) \phi^2
       \underbrace{\int_{0}^{\infty}
       e^{-\frac{z}{\lambda_2}}
       e^{-\frac{\phi^2\left(\phi P+ 1\right)}{\lambda_1 z}}
        \times z^{-2}{\rm d}z }_{
        \begin{subarray}{l}
        =2 \left(\frac{\phi^2\left(\phi P+1\right)\lambda_2}{\lambda_1}\right)^{-\frac{1}{2}}
        {\mathcal{{K}}}_1\left(2\sqrt{\frac{\phi^2\left(\phi P+1\right)}{\lambda_1\lambda_2}}\right)
%        \\
%       \leq 2 \left(\frac{\phi^2\left(\phi P+1\right)\lambda_2}{\lambda_1}\right)^{-\frac{1}{2}} \times
%       \left(2\sqrt{\frac{\phi^2\left(\phi P+1\right)}{\lambda_1\lambda_2}}\right)^{-1}
       =\frac{\lambda_1}{\phi^2 \left(\phi P+1\right)}
        \end{subarray}
        }
       {\rm d}\phi
       \\
       \label{upperBound_F21}
       & \leq D_{22} \times e^{-\frac{\beta}{\lambda_2 P}} \times \int_{\frac{\beta}{P}}^{\frac{\beta+\sqrt{\beta^2 + 8\Delta\beta}}{2P}}
       (\phi P + 1)
       {\rm d}\phi
       +
       D_{23} \times e^{-\frac{\beta}{\lambda_2 P}}
        \int_{\frac{\beta}{P}}^{\frac{\beta+\sqrt{\beta^2 + 8\Delta\beta}}{2P}}
        {\sqrt{\phi P + 1}}
       {\rm d}\phi \nonumber\\
       &
        +
        D_{24} \times e^{-\frac{\beta}{\lambda_2 P}}
        \int_{\frac{\beta}{P}}^{\frac{\beta+\sqrt{\beta^2 + 8\Delta\beta}}{2P}}
       {\rm d}\phi,
       %\\
%        & \leq D_{24} \times e^{-\frac{\beta}{\lambda_2 P}}
%        \int_{\frac{\beta}{P}}^{\frac{\beta+\sqrt{\beta^2 + 8\Delta\beta}}{2P}}
%       {\rm d}\phi = D_{24} \times e^{-\frac{\beta}{\lambda_2 P}} \times \frac{\sqrt{\beta^2 + 8\Delta\beta} - \beta}{2P}
%       \leq D_{24} \times e^{-\frac{\beta}{\lambda_2 P}} \times \frac{\sqrt{\beta^2} + \sqrt{8\Delta\beta} - \beta}{2P} \nonumber\\
%       & = D_{24} \times e^{-\frac{\beta}{\lambda_2 P}} \times \frac{\sqrt{2 \beta} \times \sqrt{\Delta}}{P}
%       \nonumber\\
%       \label{upperBound_F21}
%       & =
%        D_{25} \times
%e^{-\frac{D_{{26}}}{P}} \times \frac{1}{P} \times \sqrt{\Delta} \leq D_{25} \times
%e^{-\frac{D_{{26}}}{P}} \times \frac{1 + \sqrt{P}}{P} \times
%\max\left\{e^{-\frac{T\Delta}{\lambda_1}}, \sqrt{\Delta}, \Delta\right\},
     \end{align}
     where \eqref{explain_a0} and \eqref{explain_a1} are derived based
     on: (\rmnum{1}) $\int_0^{\infty} x^{v - 1}e^{-\frac{\beta}{x} - \gamma x}{\rm d}x =
      2\left(\frac{\beta}{\gamma}\right)^{\frac{v}{2}}{\mathcal{K}}_{v}\left(2\sqrt{\beta\gamma}\right)$ \
      \cite[ Eq. (3.471.9)]{Tables_Integrals}
      with ${\mathcal{K}}_v (z)$ being the modified bessel function of the second kind; (\rmnum{2}) ${\mathcal{K}}_0 (x)
      \leq \frac{2}{x}$ and
      ${\mathcal{K}}_{-1} (x) = {\mathcal{K}}_{1} (x) \leq \frac{1}{x}$ for $x > 0$
      \cite[ Eq. (27)]{Erdem_JSAC_Distributed_Beamforming_Relay}. After basic calculations, we obtain
      \begin{align}
      \label{f21_final_bound}
        \mathscr{F}_{2, 1} \leq D_{25} \times e^{-\frac{D_{26}}{P}} \times \frac{\Delta + \sqrt{\Delta}}{P}.
      \end{align}

      For $\mathscr{F}_{2, 2}$, because $H_1 > H_2$ and $q_o(H_1) > q_o(H_2) = \Delta \geq \geq H_2$, we have
    \begin{align}
    \label{f22_alpha}
      \alpha_{q_o} & =
    \frac{2q_o(H_2)}{\sqrt{\left[q_o(H_1)+q_o(H_2)\right]^2 + 4 q_o(H_1) q_o^2(H_2)P}+q_o(H_1)+q_o(H_2)} \nonumber\\
    & \leq
    \frac{2q_o(H_2)}{\sqrt{\left[q_o(H_2)+q_o(H_2)\right]^2 + 4 q_o(H_2) q_o^2(H_2)P}+q_o(H_2)+q_o(H_2)} \nonumber\\
    & =\frac{1}{\sqrt{q_o(H_2)P+1}+1} = \frac{1}{\sqrt{P\Delta+1}+1}.
    \end{align}
    Since $r_2\left(\alpha_{q_o}\right)$ is decreasing on $\alpha_{q_o}$, we obtain $r_2\left(\alpha_{q_o}\right) \geq
     r_2\left(\frac{1}{\sqrt{P\Delta+1}+1}\right)$ and ${\pmb 1}_{r_2\left(\alpha_{q_o}\right) < r_{\rm th}} \leq
     {\pmb 1}_{r_2\left(\frac{1}{\sqrt{P\Delta+1}+1}\right) < r_{\rm th}} ={\pmb 1}_{\frac{H_2\left(1 - \frac{1}{\sqrt{P\Delta+1}+1}\right)}{PH_2\frac{1}{\sqrt{P\Delta+1}+1} + 1} < \frac{\beta}{P}} = {\pmb 1}_{\frac{H_2\sqrt{P\Delta+1}}{PH_2+1 + \sqrt{P\Delta+1}} < \frac{\beta}{P}} \leq {\pmb 1}_{\frac{H_2\sqrt{P\Delta+1}}{P\Delta+1 + \sqrt{P\Delta+1}} < \frac{\beta}{P}} = {\pmb 1}_{H_2 \leq \frac{\beta\left(\sqrt{P\Delta+1} + 1\right)}{P}}$. Similar to \eqref{F1_upperBound}, we will have
     \begin{align}
     \label{f22}
       \mathscr{F}_{2, 2} \leq \int_{\mathpzc{I}_2} {\pmb 1}_{H_2 < \beta\frac{\sqrt{P\Delta+1}+1}{P}} \prod_{i=1}^2 f_{H_i}(H_i)
  {\rm d}
  H_i \leq D_{27} \times
e^{-\frac{D_{{28}}}{P}}
\times \frac{\sqrt{P\Delta}+1}{P} \times \Delta.
     \end{align}

      For $\mathscr{F}_3$, since $q_o(H_1) \geq q_o(H_2)$ and $q_o(H_i) - \Delta \leq H_i \leq q_o(H_i)$ for $i = 1, 2$, we obtain
      \begin{align}
          r_1\left(\alpha_{q_o}\right) & = \log_2\left(1 + P H_1 \times \alpha_{q_o}\right)
      \geq \log_2\left(1 + P \times (q_o(H_1) - \Delta) \times \alpha_{q_o}\right)   \nonumber\\
      &= \log_2\left(1 + P \times q_o(H_1) \times \alpha_{q_o} - P \times \Delta \times \alpha_{q_o}\right)\\
      &= \log_2\left(
        1 + P \times g_{\geq}(q_o(H_1), q_o(H_2)) - P\times g_{\geq}(q_o(H_1), q_o(H_2)) \times \frac{\Delta}{q_o(H_1)}
        \right) \\
        &=\log_2\left(
        1 + P \times g_{\geq}(q_o(H_1), q_o(H_2)) \times \left(1 - \frac{\Delta}{q_o(H_1)}\right)
        \right) \nonumber\\
        & \geq
        \log_2\left(
        1 + P \times g_{\geq}(q_o(H_1), q_o(H_2)) \times \left(1 - \frac{\Delta}{q_o(H_2)}\right)
        \right)\\
        &\geq
        \log_2\left(
        1 + P \times g_{\geq}(H_1, H_2) \times \left(1 - \frac{\Delta}{q_o(H_2)}\right)
        \right), \\
      r_2\left(\alpha_{q_o}\right)
      &= \log_2\left(1 + \frac{H_2 \left(1 - \alpha_{q_o}\right)}{H_2 \alpha_{q_o} + \frac{1}{P}}\right)
      = \log_2\left(1 + \frac{(q_o(H_2) - \Delta) \times
       \left(1 - \alpha_{q_o}\right)}{\left(q_o(H_2) - \Delta\right) \times \alpha_{q_o} + \frac{1}{P}}\right)\\
       & \geq   \log_2\left(1 + \frac{(q_o(H_2) - \Delta)
       \times \left(1 - \alpha_{q_o}\right)}{q_o(H_2) \times \alpha_{q_o} + \frac{1}{P}}\right)
       = \log_2\left(1 + \frac{q_o(H_2) \times \left(1 - \alpha_{q_o}\right)}{q_o(H_2) \times \alpha_{q_o} + \frac{1}{P}}
        -
        \frac{\Delta\times \left(1 - \alpha_{q_o}\right)}{q_o(H_2) \times \alpha_{q_o} + \frac{1}{P}}\right)\\
        &
        = \log_2\left(
        1 + P \times g_{\geq}(q_o(H_1), q_o(H_2)) \times \left(1 - \frac{\Delta}{q_o(H_2)}\right)
        \right)\\
        & \geq
        \log_2\left(
        1 + P \times g_{\geq}(H_1, H_2) \times \left(1 - \frac{\Delta}{q_o(H_2)}\right)
        \right),
      \end{align}
      Therefore, we have
      \begin{align}
        {\pmb 1}_{\min\left\{r_1\left(\alpha_{q_o}\right), r_2\left(\alpha_{q_o}\right)\right\} < {r}_{\rm th}}
      & \leq
        {\pmb 1}_{\log_2\left(
        1 + P \times g_{\geq}(H_1, H_2) \times \left(1 - \frac{\Delta}{q_o(H_2)}\right)
        \right) < r_{\rm th}}
         = {\pmb 1}_{g_{\geq}(H_1, H_2) < \frac{\beta}{P\left(1 - \frac{\Delta}{q_o(H_2)}\right)}}\nonumber\\
         \label{last_inequality}
         &
        \leq
        {\pmb 1}_{g_{\geq}(H_1, H_2) < \frac{\beta}{P}\left(1 + \frac{2\Delta}{q_o(H_2)}\right)}
        \leq  {\pmb 1}_{g_{\geq}(H_1, H_2) < \frac{\beta}{P}\left(1 + \frac{2\Delta}{H_2}\right)},
      \end{align}
      where \eqref{last_inequality} is because $\left(1 - \frac{\Delta}{q_o(H_2)}\right) \times \left(1 + \frac{2\Delta}{q_o(H_2)}\right)
      = 1 + \frac{\Delta}{q_o(H_2)} - 2 \left(\frac{\Delta}{q_o(H_2)}\right)^2 \geq 1$ since
      $q_o(H_2) \geq 2\Delta$ for $(H_1, H_2) \in \mathpzc{I}_3$, and $q_o(H_2) \geq H_2$. Similar to \eqref{exp03}
      and \eqref{upperBound_F21}, we can obtain an upper bound on $\mathscr{F}_3$ (the detailed derivation is omitted due to similarity).
      For $\mathscr{F}_4$, its upper bound can be developed in the same way as the upper
      bound on $\mathscr{F}_3$.

      For $\mathscr{F}_5$, when $H_1 \geq H_2\geq \Delta$, since $g_{\geq}(H_1, H_2) \geq
      \frac{2H_1 H_2}{\sqrt{(H_1 + H_1)^2 + 4H_1^2 H_2 P} + H_1 + H_1} = \frac{H_2}{\sqrt{PH_2+1}+1}$, we obtain from
      \eqref{last_inequality} that
      \begin{align}
      \label{f5_bound}
        {\pmb 1}_{\min\left\{r_1\left(\alpha_{q_o}\right), r_2\left(\alpha_{q_o}\right)\right\} < r_{\rm th}}
      & \leq  {\pmb 1}_{g_{\geq}(H_1, H_2) < \frac{\beta}{P}\left(1 + \frac{2\Delta}{H_2}\right)}
      \leq
      {\pmb 1}_{g_{\geq}(H_1, H_2) < \frac{\beta}{P}\left(1 + \frac{2\Delta}{\Delta}\right) = \frac{3\beta}{P}} \leq
        {\pmb 1}_{\frac{H_2}{\sqrt{1 + H_2 P} + 1} < \frac{3\beta}{P}}
        = {\pmb 1}_{H_2 < \frac{D_{29}}{P}},
      \end{align}
      where $D_{29} = \left(3\beta+1\right)^2-1$. Similarly, when $H_1 < H_2$,
      we have ${\pmb 1}_{\min\left\{r_1\left(\alpha_{q_o}\right),
      r_2\left(\alpha_{q_o}\right)\right\} < r_{\rm th} } \leq {\pmb 1}_{H_1 < \frac{D_{29}}{P}}$.
      Therefore, an upper bound on $\mathscr{F}_5$ is
      \begin{align}
        \mathscr{F}_5
        & \leq
        \int_{\mathpzc{I}_{4} \cap \left\{(H_1, H_2): H_1 \geq H_2\right\}}
        {\pmb 1}_{H_2 < \frac{D_{29}}{P}} \times
        \prod_{i=1}^2 f_{H_i}(H_i)
  {\rm d}
  H_i
   +
        \int_{\mathpzc{I}_{4} \cap  \left\{(H_1, H_2): H_1 < H_2\right\}}
        {\pmb 1}_{H_1 < \frac{D_{29}}{P}} \times
        \prod_{i=1}^2 f_{H_i}(H_i)
  {\rm d}
  H_i \nonumber\\
        \label{inequalityF4}
        & \leq
        \underbrace{\int_{T\Delta}^{\infty}\frac{1}{\lambda_1}e^{-\frac{H_1}{\lambda_1}} {\rm d}H_1}_{=e^{-\frac{T\Delta}{\lambda_1}}}
        \int_{\frac{\beta}{P}}^{\frac{D_{29}}{P}} \frac{1}{\lambda_2}\underbrace{e^{-\frac{H_2}{\lambda_2}}}_{\leq
        e^{-\frac{\beta}{P\lambda_2}} \leq e^{-\frac{\beta}{P\lambda_1}}} {\rm d}H_2
        +
        \underbrace{\int_{T\Delta}^{\infty}\frac{1}{\lambda_2}e^{-\frac{H_2}{\lambda_2}} {\rm d}H_2}_{=e^{-\frac{T\Delta}{\lambda_2}} \leq e^{-\frac{T\Delta}{\lambda_1}}}
        \int_{\frac{\beta}{P}}^{\frac{D_{29}}{P}} \frac{1}{\lambda_1}\underbrace{e^{-\frac{H_1}{\lambda_1}}}_{\leq
        e^{-\frac{\beta}{P\lambda_1}}} {\rm d}H_1\\
        & \leq e^{-\frac{T\Delta}{\lambda_1}} \times \frac{1}{\lambda_2} \times e^{-\frac{\beta}{P\lambda_1}} \times \frac{D_{29}-\beta}{P}
        + e^{-\frac{T\Delta}{\lambda_1}} \times \frac{1}{\lambda_1} \times e^{-\frac{\beta}{P\lambda_1}} \times \frac{D_{29}-\beta}{P} \nonumber\\
        \label{f5_final_bound}
        & \leq D_{30} \times e^{-\frac{D_{31}}{P}} \times \frac{1}{P} \times e^{-\frac{T\Delta}{\lambda_1}},
        \end{align}
        where \eqref{inequalityF4} is based on the assumption that $\lambda_1 \geq \lambda_2$. Summarizing the upper bounds
        on $\mathcal{F}_1, \ldots, \mathcal{F}_5$ in \eqref{F1_upperBound}, \eqref{f21_final_bound} and \eqref{f5_final_bound} results in
        \begin{align}
        \label{final_bound}
          \mathtt{out}_{{\rm loss}, q_{o}}
          & \leq D_{32} \times e^{-\frac{D_{33}}{P}} \times \left[\frac{\sqrt{\Delta} + \Delta + e^{-\frac{T\Delta}{\lambda_1}}}{P}
          + \frac{\Delta^{\frac{3}{2}}}{\sqrt{P}}
          \right] \nonumber\\
          & \leq D_{34} \times e^{-\frac{D_{33}}{P}} \times \frac{\sqrt{P}+1}{P} \times \max\left\{\Delta^{\frac{1}{2}},
          \Delta, \Delta^{\frac{3}{2}},  e^{-\frac{T\Delta}{\lambda_1}}\right\}
          \nonumber\\
          & =
          D_{34} \times e^{-\frac{D_{33}}{P}} \times \frac{\sqrt{P}+1}{P} \times \max\left\{\Delta^{\frac{1}{2}},
          \Delta^{\frac{3}{2}},  e^{-\frac{T\Delta}{\lambda_1}}\right\},
        \end{align}
        which completes the proof of the upper bound on $\mathtt{out}_{{\rm loss}, q_{o}}$ in \eqref{outBound}.
        \hfill$\blacksquare$

\section*{Appendix C: Proof of Lemma 4}
Similar to \eqref{f5_bound}, the full-CSI outage probability in \eqref{out_min} can be
derived as
\begin{align}
  \mathtt{out}_{\min}
  & = \int {\pmb 1}_{\min\left\{r_1\left(\alpha^{\star}\right), r_2\left(\alpha^{\star}\right)\right\} < r_{\rm th}}
  \prod_{i=1}^2 f_{H_i}(H_i)
  {\rm d}
  H_i \nonumber\\
  & = \int_{\left\{(H_1, H_2): H_1 \geq H_2\right\}} {\pmb 1}_{g_{\geq}(H_1, H_2) < \frac{\beta}{P}} \prod_{i=1}^2 f_{H_i}(H_i)
  {\rm d}
  H_i  + \int_{\left\{(H_1, H_2): H_1 < H_2\right\}} {\pmb 1}_{g_{<}(H_1, H_2) < \frac{\beta}{P}} \prod_{i=1}^2 f_{H_i}(H_i)
  {\rm d}
  H_i \nonumber\\
  & \leq \int_{\left\{(H_1, H_2): H_1 \geq H_2\right\}} {\pmb 1}_{H_2 < \frac{D_{35}}{P}} \prod_{i=1}^2 f_{H_i}(H_i)
  {\rm d}
  H_i
  +
  \int_{\left\{(H_1, H_2): H_1 < H_2\right\}} {\pmb 1}_{H_1 < \frac{D_{35}}{P}} \prod_{i=1}^2 f_{H_i}(H_i)
  {\rm d}
  H_i \nonumber\\
  \label{diversity_out}
  & \leq \int_{0}^{\frac{D_{35}}{P}}\frac{e^{-\frac{H_2}{\lambda_2}}}{\lambda_2}{\rm d}H_2
  +
  \int_{0}^{\frac{D_{35}}{P}}\frac{e^{-\frac{H_1}{\lambda_1}}}{\lambda_1}{\rm d}H_1
  = 1 - e^{-\frac{D_{35}}{P\lambda_2}} + 1 - e^{-\frac{D_{35}}{P\lambda_1}}
  \leq \frac{D_{35}\left(\frac{1}{\lambda_1} + \frac{1}{\lambda_2}\right)}{P},
\end{align}
where $D_{35} = (\beta+1)^2-1$. Thus, the diversity order is $d_{\max} =
\lim_{P\rightarrow\infty}\frac{\log \mathtt{out}_{\min}}{\log P} \geq 1$. It is
straightforward to show that $d_{\max} \leq 1$, which means $d_{\max} = 1$. Since
$\mathtt{out}_{\min} = \text{Pr}\left\{r_1\left(\alpha^{\star}\right) < r_{\rm th}\right\} =
\text{Pr}\left\{r_2\left(\alpha^{\star}\right) < r_{\rm th}\right\}$, the maximum
achievable diversity order for both receivers is $1$.

When $q_{o}(\cdot)$ is employed, the outage probability of Receiver $i$ is
$\mathtt{out}_{q_{o}, i} = \int {\pmb 1}_{r_i\left(\alpha_{q_o}\right) < r_{\rm th}}
\prod_{i=1}^2 f_{H_i}(H_i)
  {\rm d}
  H_i$ for $i = 1, 2$. Similar to the derivations of $\mathscr{F}_i$ for $i = 1, \ldots, 5$ in \eqref{F1_upperBound}, \eqref{f21_final_bound} and \eqref{f5_final_bound}, we will obtain
  $\mathtt{out}_{q_{o}, 1} \leq \mathtt{out}_{\min} + D_{36} \times e^{-\frac{D_{37}}{P}} \times \left[\frac{\sqrt{\Delta} +e^{-\frac{T\Delta}{\lambda_1}}}{P}
          + \frac{\Delta^{\frac{3}{2}}}{\sqrt{P}}
          \right]$ and $\mathtt{out}_{q_{o}, 2} \leq \mathtt{out}_{\min} + D_{38} \times e^{-\frac{D_{39}}{P}} \times \frac{D_{40} + \Delta +
          e^{-\frac{T\Delta}{\lambda_1}}}{P}$.\footnote{Note that when we derive the diversity order for $\mathscr{F}_{2, 2}$, we will not use the upper bound in \eqref{f22}. We can further obtain from \eqref{f22_alpha} that $\alpha_{q_o} \leq \frac{1}{\sqrt{P\Delta+1}+1} \leq \frac{1}{\sqrt{PH_2+1}+1}$, then, ${\pmb 1}_{r_2\left(\alpha_{q_o}\right) < r_{\rm th}} \leq {\pmb 1}_{r_2\left(\frac{1}{\sqrt{PH_2+1}+1}\right) < r_{\rm th}} = {\pmb 1}_{H_2 < \frac{\beta^2 +\beta}{P}}$, and it is trivial to obtain that $\mathscr{F}_{2, 2} \leq  D_{41} \times \frac{e^{-\frac{D_{42}}{P}}}{P}$.} Therefore, for fixed $\Delta$, the diversity orders of $1/2$ and $1$ are achievable for
          Receivers 1 and 2, respectively.

  For Receiver 1, $\frac{\Delta^{\frac{3}{2}}}{\sqrt{P}}$ in the upper bound on $\mathtt{out}_{q_{o}, 1}$
  is the bottleneck for diversity gains.
  If we scale $\Delta$ as $\Delta^{\frac{3}{2}}\sim_P \frac{1}{\sqrt{P}}$,
    i.e., $\Delta\sim_P P^{-\frac{1}{3}}$, the diversity order of $1$ is also achievable for Receiver 1.

\hfill$\blacksquare$

\section*{Appendix D: Proof of Lemma 5}

Given $K$ and $\beta > 0$, define the following two optimization problems:
\begin{enumerate}
\item[(\textbf{P1})] $r^{\star}_{\max}(K, \beta) =
\max\limits_{{\pmb \alpha} = \left[\alpha_1, \ldots,\alpha_K\right]} \min\limits_{k = 1,
\ldots, K} r_k ({\pmb \alpha})$, subject to $0 \leq \alpha_k\leq \beta$ and
$\sum_{k=1}^K \alpha_k = \beta$.
\item[(\textbf{P2})] ${r}^{\dag}_{\max}(K, \beta) =
\max\limits_{{\pmb \alpha} = \left[\alpha_1, \ldots,\alpha_K\right]} \min\limits_{k = 1,
\ldots, K} r_k ({\pmb \alpha})$, subject to $r_1 ({\pmb \alpha}) = \cdots = r_K ({\pmb \alpha})$, $0 \leq \alpha_k\leq \beta$, and $\sum_{k=1}^K \alpha_k = \beta$,
\end{enumerate}
where (\textbf{P1}) is the original optimization problem in \eqref{kuser} when $\beta = 1$. We will show that the maximum minimum rates of (\textbf{P1}) and (\textbf{P2}) are the same, i.e., $r^{\star}_{\max}(K, \beta) = {r}^{\dag}_{\max}(K, \beta)$, which proves the lemma.

Denote the optimal power allocations for (\textbf{P1}) and (\textbf{P2}) by ${\pmb \alpha}^{\star}_{K}(\beta) = \left[{\alpha}^{\star}_{1, K}(\beta), \ldots, {\alpha}^{\star}_{K, K}(\beta)\right]$ and ${\pmb \alpha}^{\dag}_{K}(\beta) = \left[{\alpha}^{\dag}_{1, K}(\beta), \ldots, {\alpha}^{\dag}_{K, K}(\beta)\right]$, respectively. Since $r^{\star}_{\max}(K, \beta) \geq {r}^{\dag}_{\max}(K, \beta)$, it is sufficient to prove that $r^{\star}_{\max}(K, \beta) \leq {r}^{\dag}_{\max}(K, \beta)$.

The proof for $K = 2$ is provided in the proof of Theorem 2. By induction,
assume $r^{\star}_{\max}(K, \beta) = {r}^{\dag}_{\max}(K, \beta)$ holds for $K = K_1$. When
$K = K_1 + 1$, there are two possibilities:
\begin{itemize}
\item[(\rmnum{1})] If $r_{K_1+1}\left({\pmb \alpha}^{\star}_{K_1 + 1}(\beta)\right) \geq r_{K_1+1}
    \left({\pmb \alpha}^{\dag}_{K_1 + 1}(\beta)\right)$, since $r_{K_1+1}\left({\pmb \alpha}\right) = \log_2\left(1 + \frac{\alpha_{K_1+1}}{\sum_{i=1}^{K_1} \alpha_i + \frac{1}{PH_{K_1+1}}}\right) = \log_2\left(1 + \frac{\alpha_{K_1+1}}{\beta - \alpha_{K_1+1} + \frac{1}{PH_{K_1+1}}}\right)$ for any $\pmb \alpha$ satisfying $\sum_{i=1}^{K_1+1}\alpha_i = \beta$, it must have
    ${\alpha}^{\star}_{K_1+1, K_1+1}(\beta) \geq {\alpha}^{\dag}_{K_1+1, K_1+1}(\beta)$,
    then, $\beta_1 = \sum_{k=1}^{K_1} {\alpha}^{\star}_{k, K_1+1}(\beta) = \beta -
    {\alpha}^{\star}_{K_1+1, K_1+1}(\beta) \leq \beta - {\alpha}^{\dag}_{K_1+1,
    K_1+1}(\beta) = \sum_{k = 1}^{K_1}{\alpha}^{\dag}_{k, K_1+1}(\beta) = \beta_2$.
    Next, we obtain
\begin{align}
r^{\star}_{\max}(K_1 + 1, \beta) & = \min\left\{ \left\{\min_{k=1, \ldots, K_1} r_k\left({\pmb \alpha}^{\star}_{K_1 + 1}(\beta)\right)\right\}, r_{K_1+1}\left({\pmb \alpha}^{\star}_{K_1 + 1}(\beta)\right)\right\} \nonumber\\
\label{explain_01101}
& \leq \min\left\{r^{\star}_{\max}\left(K_1, \beta_1\right), r_{K_1+1}\left({\pmb \alpha}^{\star}_{K_1 + 1}(\beta)\right)\right\}
\\
\label{explain_01102}
& = \min\left\{{r}^{\dag}_{\max}\left(K_1, \beta_1\right), r_{K_1+1}\left({\pmb \alpha}^{\star}_{K_1 + 1}(\beta)\right)\right\}
\\
\label{explain_01103}
& \leq \min\left\{{r}^{\dag}_{\max}\left(K_1, \beta_2\right), r_{K_1+1}\left({\pmb \alpha}^{\star}_{K_1 + 1}(\beta)\right)\right\}
\\
\label{explain_01104}
& = \min\left\{{r}^{\dag}_{\max}\left(K_1+1, \beta\right), r_{K_1+1}\left({\pmb \alpha}^{\star}_{K_1 + 1}(\beta)\right)\right\} \\
& = \min\left\{r_{K_1+1}\left({\pmb \alpha}^{\dag}_{K_1 + 1}(\beta)\right), r_{K_1+1}\left({\pmb \alpha}^{\star}_{K_1 + 1}(\beta)\right)\right\} \nonumber\\
& = r_{K_1+1}\left({\pmb \alpha}^{\dag}_{K_1 + 1}(\beta)\right) = {r}^{\dag}_{\max}\left(K_1+1, \beta\right). \nonumber
\end{align}
Thus, $r^{\star}_{\max}(K_1 + 1, \beta) \leq {r}^{\dag}_{\max}(K_1 + 1, \beta)$. The inequality
\eqref{explain_01101} is due to the optimality of $r^{\star}_{\max}\left(K_1, \beta_1\right)$;
\eqref{explain_01102} arises from the assumption that $r^{\star}_{\max}(K, \beta_1) =
{r}^{\dag}_{\max}(K, \beta_1)$ when $K = K_1$; \eqref{explain_01103} is because
${r}^{\dag}_{\max}(K, \beta)$ is non-decreasing on $\beta$; \eqref{explain_01104} holds since ${r}^{\dag}_{\max}\left(K_1, \beta_2\right) = {r}^{\dag}_{\max}\left(K_1+1, \beta\right)$.
\item[(\rmnum{2})] If $r_{K_1+1}\left({\pmb \alpha}^{\star}_{K_1 + 1}(\beta)\right) < r_{K_1+1}\left({\pmb \alpha}^{\dag}_{K_1 + 1}(\beta)\right)$,
we have $r^{\star}_{\max}(K_1+1, \beta) \leq r_{K_1+1}\left({\pmb \alpha}^{\star}_{K_1 + 1}(\beta)\right)
< r_{K_1+1}\left({\pmb \alpha}^{\dag}_{K_1 + 1}(\beta)\right) = {r}^{\dag}_{\max}(K_1+1, \beta)$, which completes the proof of Lemma 5. \hfill$\blacksquare$
\end{itemize}

\ifCLASSOPTIONcaptionsoff
  \newpage
\fi

\bibliographystyle{IEEEtran}
\bibliography{IEEEabrv,Monkey}

\begin{thebibliography}{10}
\providecommand{\url}[1]{#1}
\csname url@samestyle\endcsname
\providecommand{\newblock}{\relax}
\providecommand{\bibinfo}[2]{#2}
\providecommand{\BIBentrySTDinterwordspacing}{\spaceskip=0pt\relax}
\providecommand{\BIBentryALTinterwordstretchfactor}{4}
\providecommand{\BIBentryALTinterwordspacing}{\spaceskip=\fontdimen2\font plus
\BIBentryALTinterwordstretchfactor\fontdimen3\font minus
  \fontdimen4\font\relax}
\providecommand{\BIBforeignlanguage}[2]{{%
\expandafter\ifx\csname l@#1\endcsname\relax
\typeout{** WARNING: IEEEtran.bst: No hyphenation pattern has been}%
\typeout{** loaded for the language `#1'. Using the pattern for}%
\typeout{** the default language instead.}%
\else
\language=\csname l@#1\endcsname
\fi
#2}}
\providecommand{\BIBdecl}{\relax}
\BIBdecl

\bibitem{SystemLevelNOMA}
Y.~Saito, A.~Benjebbour, Y.~Kishiyama, and T.~Nakamura, ``System-level
  performance evaluation of downlink non-orthogonal multiple access ({NOMA}),''
  in \emph{IEEE 24th Annual International Symposium on Personal, Indoor, and
  Mobile Radio Communications (PIMRC)}, Sept. 2013, pp. 611--615.

\bibitem{3gppNOMA}
{3rd Generation Partnership Project (3GPP)}, ``Study on downlink multiuser
  superposition transmission for {LTE},'' Mar. 2015.

\bibitem{NOMA_Info}
P.~Xu, Z.~Ding, X.~Dai, and H.~V. Poor, ``A new evaluation criterion for
  non-orthogonal multiple access in 5{G} software defined networks,''
  \emph{IEEE Access}, vol.~3, pp. 1633--1639, 2015.

\bibitem{MIMONOMA}
Z.~Ding, R.~Schober, and H.~V. Poor, ``A general {MIMO} framework for noma
  downlink and uplink transmission based on signal alignment,'' \emph{{IEEE}
  Trans. Wireless Commun.}, vol.~15, no.~6, pp. 4438--4454, June 2016.

\bibitem{NOMA_Fair}
J.~A. Oviedo and H.~R. Sadjadpour, ``A new {NOMA} approach for fair power
  allocation,'' in \emph{IEEE Conference on Computer Communications Workshops
  (INFOCOM WKSHPS)}, Apr. 2016, pp. 843--847.

\bibitem{PA_NOMA}
Z.~Yang, Z.~Ding, P.~Fan, and N.~Al-Dhahir, ``A general power allocation scheme
  to guarantee quality of service in downlink and uplink {NOMA} systems,''
  \emph{{IEEE} Trans. Wireless Commun.}, vol.~15, no.~11, pp. 7244--7257, Nov.
  2016.

\bibitem{JChoi_NOMA_PA}
J.~Choi, ``Power allocation for max-sum rate and max-min rate proportional
  fairness in noma,'' \emph{{IEEE} Commun. Lett.}, vol.~20, no.~10, pp.
  2055--2058, Oct. 2016.

\bibitem{TransmitPowerMinimizationNOMA}
L.~Lei, D.~Yuan, and P.~Varbrand, ``On power minimization for non-orthogonal
  multiple access ({NOMA}),'' \emph{{IEEE} Commun. Lett.}, vol.~20, no.~12, pp.
  2458--2461, Dec. 2016.

\bibitem{Erdem_IT_Distributed_Beamforming_Relay_Interference}
E.~Koyuncu and H.~Jafarkhani, ``Distributed beamforming in wireless multiuser
  relay-interference networks with quantized feedback,'' \emph{{IEEE} Trans.
  Inf. Theory}, vol.~58, no.~7, pp. 4538--4576, July 2012.

\bibitem{Cooperative_Quantization}
X.~Liu, E.~Koyuncu, and H.~Jafarkhani, ``Cooperative quantization for two-user
  interference channels,'' \emph{{IEEE} Trans. Commun.}, vol.~63, no.~7, pp.
  2698--2712, 2015.

\bibitem{noma_feedback}
S.~Liu and C.~Zhang, ``Downlink non-orthogonal multiple access system with
  limited feedback channel,'' in \emph{International Conference on Wireless
  Communications Signal Processing (WCSP)}, Oct. 2015, pp. 1--5.

\bibitem{MassiveMIMOFeedback}
Z.~Ding and H.~V. Poor, ``Design of {Massive-MIMO-NOMA} with limited
  feedback,'' \emph{{IEEE} Signal Process. Lett.}, vol.~23, no.~5, pp.
  629--633, May 2016.

\bibitem{OneBitFeedback}
P.~Xu, Y.~Yuan, Z.~Ding, X.~Dai, and R.~Schober, ``On the outage performance of
  non-orthogonal multiple access with 1-bit feedback,'' \emph{{IEEE} Trans.
  Wireless Commun.}, vol.~15, no.~10, pp. 6716--6730, Oct. 2016.

\bibitem{NOMAOutageConstraint}
J.~Cui, Z.~Ding, and P.~Fan, ``A novel power allocation scheme under outage
  constraints in {NOMA} systems,'' \emph{{IEEE} Signal Process. Lett.},
  vol.~23, no.~9, pp. 1226--1230, Sept. 2016.

\bibitem{ImperfectCSINOMA}
Z.~Yang, Z.~Ding, P.~Fan, and G.~K. Karagiannidis, ``On the performance of
  non-orthogonal multiple access systems with partial channel information,''
  \emph{{IEEE} Trans. Commun.}, vol.~64, no.~2, pp. 654--667, Feb. 2016.

\bibitem{David_Love_Grassmanian}
D.~J. Love, {R. W. Heath, Jr.}, and T.~Strohmer, ``Grassmannian beamforming for
  multiple-input multiple-output wireless systems,'' \emph{\textit{IEEE Trans.
  Inf. Theory}}, vol.~49, no.~10, pp. 2735--2747, Oct. 2003.

\bibitem{Xiaoyi_Multicast_WCOM}
X.~Liu, E.~Koyuncu, and H.~Jafarkhani, ``Multicast networks with
  variable-length limited feedback,'' \emph{\textit{IEEE Trans. Wireless
  Commun.}}, vol.~14, no.~1, pp. 252--264, Jan. 2015.

\bibitem{NOMA_Power_Allocation}
J.~Choi, ``On the power allocation for a practical multiuser superposition
  scheme in {NOMA} systems,'' \emph{{IEEE} Commun. Lett.}, vol.~20, no.~3, pp.
  438--441, Mar. 2016.

\bibitem{NOMA_MaxMin_Power_Allocation}
------, ``Power allocation for max-sum rate and max-min rate proportional
  fairness in {NOMA},'' \emph{{IEEE} Commun. Lett.}, vol.~20, no.~10, pp.
  2055--2058, Oct. 2016.

\bibitem{Elements_IT}
T.~M. Cover and J.~A. Thomas, \emph{Elements of Information Theory (Wiley
  Series in Telecommunications and Signal Processing)}.\hskip 1em plus 0.5em
  minus 0.4em\relax Wiley-Interscience, 2006.

\bibitem{MaxMinFairness}
R.~Sun, M.~Hong, and Z.-Q. Luo, ``Joint downlink base station association and
  power control for max-min fairness: Computation and complexity,''
  \emph{\textit{IEEE J. Select. Areas Commun.}}, vol.~33, no.~6, pp.
  1040--1054, June 2015.

\bibitem{Erdem_VLQ_IT}
E.~Koyuncu and H.~Jafarkhani, ``Variable-length limited feedback beamforming in
  multiple-antenna fading channels,'' \emph{\textit{IEEE Trans. Inf. Theory}},
  vol.~60, no.~11, pp. 7140--7164, Nov. 2014.

\bibitem{spaceTimeCodingBook}
H.~Jafarkhani, \emph{Space-Time Coding: Theory and Practice}, 1st~ed.\hskip 1em
  plus 0.5em minus 0.4em\relax New York, NY, USA: Cambridge University Press,
  2010.

\bibitem{excessPower}
Z.~Chen, Z.~Ding, X.~Dai, and R.~Zhang, ``A mathematical proof of the
  superiority of {NOMA} compared to conventional {OMA},''
  \url{https://arxiv.org/pdf/1612.01069.pdf}, 2016.

\bibitem{Tables_Integrals}
I.~Gradshteyn and I.~Ryzhik, \emph{Table of Integrals, Series, and Products},
  7th~ed., A.~Jeffrey and D.~Zwillinger, Eds.\hskip 1em plus 0.5em minus
  0.4em\relax Academic Press, Mar. 2007.

\bibitem{Erdem_JSAC_Distributed_Beamforming_Relay}
E.~Koyuncu, Y.~Jing, and H.~Jafarkhani, ``Distributed beamforming in wireless
  relay networks with quantized feedback,'' \emph{\textit{IEEE J. Select. Areas
  Commun.}}, vol.~26, no.~8, pp. 1429--1439, Oct. 2008.

\end{thebibliography}

\end{document}